\documentclass[preprint,11pt,numbers,authoryear,preprint]{elsarticle}
\journal{Applied Energy}

\biboptions{}

\usepackage{geometry}
\newcommand{\Nd}{\ensuremath{\mathds{N}}}

\newcommand{\Zd}{\ensuremath{\mathds{Z}}}

\newcommand{\Kc}{\ensuremath{\mathcal{K}}}

\newcommand{\Oc}{\ensuremath{\mathcal{O}}}

\newcommand{\Tc}{\ensuremath{\mathcal{T}}}

\newcommand{\Xc}{\ensuremath{\mathcal{X}}}

\usepackage{amssymb, dsfont}
\usepackage{amsmath, amsthm}
\usepackage{algpseudocode}
\usepackage{algorithm}
\usepackage{graphicx}
\usepackage{subcaption}  
\usepackage{url}
\usepackage{hyperref}
\usepackage{tikz}
\usepackage{booktabs}
\usepackage{xcolor}
\usepackage{lmodern,textcomp}
\usepackage{adjustbox}
\usepackage{xcolor}
\usepackage{multirow, makecell}
\usepackage{nicefrac}
\usepackage{relsize}
\usepackage{longtable}
\usetikzlibrary{backgrounds}
\usetikzlibrary{shadows}

\theoremstyle{definition}

\newtheorem{prop}{Proposition}

\usetikzlibrary{decorations.pathreplacing,arrows.meta,calc,positioning}

\definecolor{customRed}{RGB}{53,98,111}
\definecolor{customGreen}{RGB}{229,84,81}

\DeclareMathOperator*{\argmax}{arg\,max}
\newcommand{\set}[1]{ \ensuremath{\left\{ #1 \right\}} }

\definecolor{box_color}{RGB}{31, 119, 180}
\definecolor{bridge_red}{RGB}{214, 39, 40}
\definecolor{arrow_color}{RGB}{0, 80, 155}

\usepackage[]{lineno} 

\begin{document}

\begin{frontmatter}



\title{High-frequency intraday trading for battery
storages} 

\tnotetext[tn_pub]{\copyright~2026. Published by Elsevier Ltd.\\
This is an open access article under the CC BY 4.0 license (\url{https://creativecommons.org/licenses/by/4.0/}).\\
Published in \textit{Applied Energy}: \url{https://doi.org/10.1016/j.apenergy.2026.128550}}

\author[label1]{David Schaurecker}
\ead{dschaurecker@gmail.com}
\author[label2]{David Wozabal}
\author[label3]{Nils Löhndorf}
\author[label1,label4]{Thorsten Staake}

\affiliation[label1]{organization={ETH Zurich},
            addressline={Rämistrasse 101},
            postcode={8092 Zurich},
            country={Switzerland}}

\affiliation[label2]{organization={Vrije Universiteit Amsterdam},
            addressline={De Boelelaan 1105},
            postcode={1081 HV Amsterdam},
            country={Netherlands}}

\affiliation[label3]{organization={University of Luxembourg},
            addressline={6 Rue Richard Coudenhove-Kalergi},
            postcode={L-1359},
            country={Luxembourg}}

\affiliation[label4]{organization={University of Bamberg},
            addressline={An der Weberei 5},
            postcode={96047 Bamberg},
            country={Germany}}

\begin{abstract}
    Maximizing revenue for grid-scale battery energy storage systems in continuous intraday electricity markets requires strategies that are able to seize trading opportunities as soon as new information arrives. This paper introduces and evaluates a computationally efficient, high-frequency implementation of the rolling intrinsic trading strategy for battery energy storage systems on the intraday market for power. By combining the established rolling intrinsic logic with a full limit order book representation and a fast dynamic programming approximation, our method explicitly considers the continuously updated list of buy and sell offers, market rules, a linear approximation of degradation, and other technical parameters at a millisecond resolution. The standard rolling intrinsic strategy is adapted for continuous intraday electricity markets and solved using a dynamic programming approximation that is two to three orders of magnitude faster than an exact mixed-integer linear programming solution. A detailed backtest over a full year of German order book data demonstrates that the proposed dynamic programming formulation does not reduce trading profits and enables the policy to react to every relevant order book update, enabling realistic rapid backtesting. Our results show the significant revenue potential of high-frequency trading: our policy earns $58$\% more than when re-optimizing only once every hour and $14$\% more than when re-optimizing once per minute, highlighting that profits critically depend on trading speed. Furthermore, we leverage the speed of our algorithm to train a parametric extension of the rolling intrinsic, increasing yearly revenue by $8.4$\% out of sample. 
\end{abstract}


\begin{keyword}
electricity trading \sep asset-backed trading \sep short-term power markets \sep dynamic programming \sep limit order book


\end{keyword}

\end{frontmatter}
\newpage

\section{Introduction} \label{sec:intro}

Grid-scale battery energy storage systems (BESS) have emerged as a promising solution to address short-term variability in renewable electricity production. Batteries store excess energy generated during periods of high production and release it when generation is low, providing much-needed flexibility and stabilizing the grid \citep{barbeito_grid_stable, Ullah_grid_feq}. In addition, they help manage peak demand, minimize reliance on fossil fuel peaker plants, and boost overall energy system efficiency. The rapid increase in grid-scale storage installations underscores this potential \citep{saldarini_bess, irena2017electricity}, with the projected adoption of grid-scale BESS to increase significantly by 2030.

In liberalized electricity markets, the evaluation of viable business models for BESS is crucial to this transformation. The economic feasibility and revenue streams of these systems depend on their ability to participate in various short-term energy markets, engaging in energy arbitrage, frequency regulation, and demand response. Understanding the potential of these business models helps investors and legislators make the right decisions for a sustainable energy future and increases the adoption of storage technologies. In particular, informed trading decisions for intraday (ID) trading become increasingly relevant, with ID trading volumes growing each year \citep{koch_2019} and market participants shifting their trading towards short-term markets.

Finding good trading strategies requires extensive backtesting over long time horizons. However, backtesting the continuous ID market is computationally expensive. Its millisecond frequency generates millions of decision points, making long-term simulations challenging.

This paper introduces a highly realistic, high-frequency implementation of the rolling intrinsic trading strategy for battery energy storage systems. By combining the rolling intrinsic logic with a full limit order book representation and a fast dynamic programming approximation, this approach integrates the full granular evolution of the market with exchange regulations and physical storage constraints. As we argue below, literature on the subject is limited, as most authors adopt a rather stylized view of the intraday market, avoiding most of its complexity and thereby yielding suboptimal strategies and a biased picture of revenue potentials.

Most relevant related literature on coordinated bidding in electricity spot markets focuses on the day-ahead market, such as \citet{fk08}, or treats the intraday market (IDM) as a single trading decision stage \citep{ff11, lwm13, kraft_23}. Applications range from pure trading strategies to asset-focused approaches, such as for hydro plant or battery system operation. Some publications consider a small number of repeated trading decisions on intraday markets for storages \citep{lohndorf_value_2023} and other applications \citep{agp16, RaWo18, kuppelwieser2023intraday}.

Interest in short-term trading strategies for grid-scale BESS has surged recently, a topic explored by, e.g., \citet{Sage_2024, cornejo2025evaluatingimpactmodelaccuracy, zhang2025jointbiddingintradayfrequency}.

Literature that exclusively focuses on (automated) trading strategies for the intraday market, while capturing its full complexity, is scarce. Some recent work presents approaches for single-, or multi-market trading strategies for storage. \citet{jp15} study short-term trading for a battery storage on real-time markets as a Markov Decision Process. \citet{brauer_2019}, and more recently, \citet{seifert24}, present a European multi-market trading strategy for battery systems across three markets, trading at discrete time intervals and connecting the sequential reserve, day-ahead and intraday market decision stages. With a 4-hour time resolution, their ability to react to short-term changes in prices is very limited. Furthermore, the resulting smoothing of price spikes significantly reduces the profits a battery could extract from intraday markets. While co-optimizing storage assets across multiple markets (e.g., day-ahead, reserve, and intraday) is highly relevant in practice, integrating the continuous intraday market into multi-market frameworks currently requires heavy simplifications, such as artificially discretizing time into fixed decision stages. These simplifications severely distort the true profit potential of high-frequency trading. We focus exclusively on capturing the full complexity of the continuous intraday market because its rapidly growing volumes and high-frequency dynamics demand specialized, highly reactive algorithms. By isolating this market, we provide a highly accurate, fast-executing benchmark standalone model. We view understanding and accurately modeling this specific market as a critical foundational step before such algorithms can be expanded into broader multi-market environments without sacrificing intraday realism. Furthermore, sophisticated operation within the continuous intraday market is becoming increasingly critical for battery storage systems, as Europe's ancillary markets face saturation from the surging volume of installed battery capacity.

The myopic rolling intrinsic (RI) trading strategy is nearly optimal for gas storage operation \citep{lohndorf2021gas}. The strategy naturally translates to other storage types, like batteries. Consequently, the RI is frequently used as a benchmark for more sophisticated strategies \citep[e.g.,][]{bertrand_2020, boukas_2021}.

Recent work by \citet{semmelmann_2024} is most closely related to our work by presenting an evaluation of RI trading for a battery on the continuous intraday market. However, their study relies on aggregated transaction data rather than the full limit order book, missing critical aspects of the market's microstructure, and does not account for battery degradation costs. While we adopt a simplified linear degradation model in this paper to maintain computational tractability at high execution frequencies, explicitly incorporating these costs into the optimization maps trading decisions more realistically and prevents the over-utilization of the asset.

Generally speaking, there are only a handful of papers that model the intraday market in its full complexity and implement trading strategies that take into account the complete information in the order book and, at the same time, do not artificially discretize time. The few examples of such papers known to the authors include \cite{bertrand_2020, boukas_2021, kuppelwieser2023intraday}. However, all of these papers are based on trading rules that are simple to execute. To the best of our knowledge, there is no single study of the rolling intrinsic policy or any other policy requiring a complex optimization for every single change in the limit order book while modeling detailed market behavior and rules.

This paper contributes to closing this gap by making the following main contributions:
\begin{enumerate}
    \item We adapt the rolling intrinsic strategy to account for continuous trading on intraday markets for electricity, explicitly modeling every single order in the limit order book. The optimization is formulated as a mixed integer linear programming problem, taking into account all relevant aspects of the problem, including order placement costs, a simplified linear battery degradation model, and detailed market rules.

    \item We approximate the exact MILP formulation using a dynamic programming approach, effectively trading a marginal loss in theoretical precision for a computational speedup of two to three orders of magnitude. As a result, our work enables the first high-frequency algorithmic trading strategy for storage assets, enabling realistic trading at every relevant update of the order book. Furthermore, we provide a detailed comparison, in terms of speed and profit, of our method against the MILP solution.

    \item In a numerical study, we conduct a detailed backtest of our strategy in the German intraday market over a full year using an order-by-order traversal of historical order book data. The results show the revenue potential of a BESS operating on the continuous intraday market. 

    The computational speed of our proposed algorithm is essential for such a detailed analysis, and our findings demonstrate that the speed of trading is critical to maximizing profits. Slower strategies, or generally less frequent trading, generate significantly lower profits by missing numerous trading opportunities due to their inability to make decisions at every relevant point in time. Our high-frequency strategy traverses the full year 2021 in around 86 minutes, solving the intrinsic optimization approximately 24 million times (4.6 solves per millisecond), submitting around 30k orders at the exchange. Furthermore, we evaluate the robustness of our strategy under optimization and battery parameter settings, providing a deeper understanding of the method's performance and generalizability.

    \item We parametrize the standard formulation of the RI and show that our simulation speed can be used to easily find optimal parameters, generating additional profits of $8.5$ \%, by slightly nudging the RI's trading behavior at no additional risk.

    \item Lastly, we publish an easy-to-use Python package to run RI simulations over extended periods of time, given a set of battery and dynamic programming parameters. This will allow other researchers or industry to realistically simulate high-frequency rolling intrinsic trading scenarios for a battery on the intraday market, thereby creating a strong benchmark that can be evaluated with minimal effort.

\end{enumerate}

The remainder of this paper is organized as follows: Section \ref{sec:power_market} offers an overview of the European Power Market, setting the stage for our proposed trading strategy, detailed in Section \ref{sec:method}. We present our results in Section \ref{sec:results}, and discuss the details and implications of our findings. Finally, Section \ref{sec:conclusion} concludes the paper and outlines directions for future research. 

\section{\label{sec:power_market}The European Power Market}
The European short-term electricity markets are served by two primary competing power exchanges, Nord Pool and EPEX SPOT. While Nord Pool originated in the Nordic region, both exchanges now operate across wide areas of Europe, including the Nordic and Central Western European (CWE) regions.

The spot markets, which include day-ahead and intraday markets, are key components of the European Single Day-Ahead Coupling (SDAC) and Single Intraday Coupling (SIDC) initiatives. These initiatives have unified market operations across Europe, enabling bids submitted in one bidding zone to be seamlessly integrated into a shared order book and traded across other zones, provided there is sufficient cross-border transmission capacity. Price differences between bidding zones are thus  primarily the result of constraints in interconnectors.

In this paper, we mainly focus on the German market, which is the largest electricity market in Europe and where trading is conducted mainly through the EPEX exchange. However, due to the ongoing convergence of the European electricity market designs, the proposed methods directly carry over to most other European markets. 

The cascading range of future markets in Germany spans from long-term markets to day-ahead and intraday markets trading hourly and subhourly contracts. In continuous intraday trading, $60$, $30$ and $15$ minute products are trading up to $5$ minutes before delivery.

\subsection{\label{ssec:intraday_market} The Intraday Market for Power}
The primary objective of short-term power market designs is to minimize system imbalances. The intraday (ID) market plays a crucial role as the final opportunity for market participants to adjust their positions in response to unforeseen changes in production and demand, often driven by updated weather forecasts. Participants aim to minimize contractual deviations and reduce potential balancing costs charged by the transmission system operator (TSO). Significant price spreads between the day-ahead and intraday markets create strategic opportunities. Participants can withhold volumes from the day-ahead market to trade them intraday for additional profit \citep{lohndorf_value_2023, seifert24}.

Intraday market designs in Europe vary by country, although there has been significant progress toward a more unified ID market structure through the Single Intraday Coupling (SIDC) initiative, which includes both continuous and recently also auction-based trading products \citep{epex2024idas}. In this paper we focus on the former. While intraday auctions clear accumulated orders at specific, discrete times at a uniform price, the continuous market matches incoming orders instantaneously against the limit order book, as discussed in detail in \citet{Neuhoff2016}. 

Unlike the day-ahead market, which clears via a single daily auction and requires only one discrete trading decision per day, the continuous intraday market operates through continuous matching and requires frequent, high-speed rebalancing. While the continuous ID market typically exhibits lower liquidity at any given point in time compared to the DA auction, it features significantly higher price volatility. This combination of unpredictability and continuous trading makes the ID market structurally more challenging to navigate, but it also creates highly rewarding arbitrage opportunities for fast-responding flexibility providers like battery storage systems.

In Germany, the intraday continuous market opens in the afternoon following the clearing of the day-ahead market. Trading begins at 3 p.m.~on the day before delivery, allowing a maximum of 32 hourly or 128 quarter-hourly products to be traded at any given time. Trading for products in the shared SIDC order book closes $60$ minutes before delivery, while products listed solely in the German order book trade until $30$ minutes before delivery. At that point, the market splits into the four German TSO regions, in each of which trading continues until $5$ minutes before delivery.

Limit orders submitted to the exchange are either immediately matched with corresponding buy or sell orders or, if unmatched, are stored in the limit order book \citep[see][]{graf_frequent_2018}. A limit order is defined by its price, volume, validity, and order type (buy/sell). Figure \ref{fig:id_price} visualizes the process of a typical submission of a new limit order to the limit order book (LOB) of a product, i.e., delivery hour.

\begin{figure}[t]
    \centering
    \input{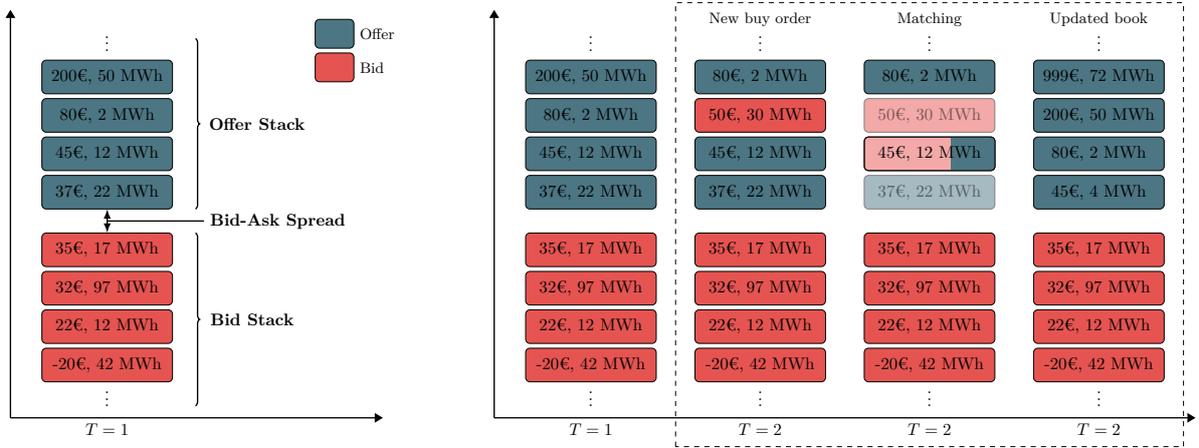}
    \caption{\label{fig:id_price} An exemplary state of the LOB at $T = 1$ presented in the left panel. The right panel depicts the clearing of a new hypothetical buy order: A new buy order with a price of 50 \texteuro/MWh, which is higher than the lowest ask price, is added to the book at $T = 2$. The quantity of this buy order (30 MWh) is then cleared against the cheapest possible offers until either the whole order is fulfilled (as is the case in the figure) or there are no offers with lower prices left. In this example, $22$ MWh out of $30$ MWh are cleared against the sell order with price 37 \texteuro/MWh \ and the remaining 30 MWh $- \ 22$ MWh $= 8$ MWh are cleared against the sell order with price 45 \texteuro/MWh. The remaining quantity of 12 MWh $- \ 8$ MWh $= 4$ MWh of the latter order stays in the order book. Note that the clearing is instantaneous, i.e., columns 2–4 in the right panel are purely illustrative and do not correspond to market states that can be observed by traders. Figure and caption adapted from \cite{graf_frequent_2018}.}
\end{figure}

Continuous markets enable participants to continuously update and optimize their orders, allowing for real-time adjustments to evolving market conditions and reactions to forecast-updates. This flexibility makes them a crucial element in automated short-term trading strategies, where a rapid response to price changes is essential. In addition, continuous markets offer traders the opportunity to exploit short-lived market inefficiencies, enhancing their ability to capitalize on arbitrage opportunities throughout the trading day.

\begin{figure}[t]
    \centering
    \includegraphics[width=0.6\textwidth]{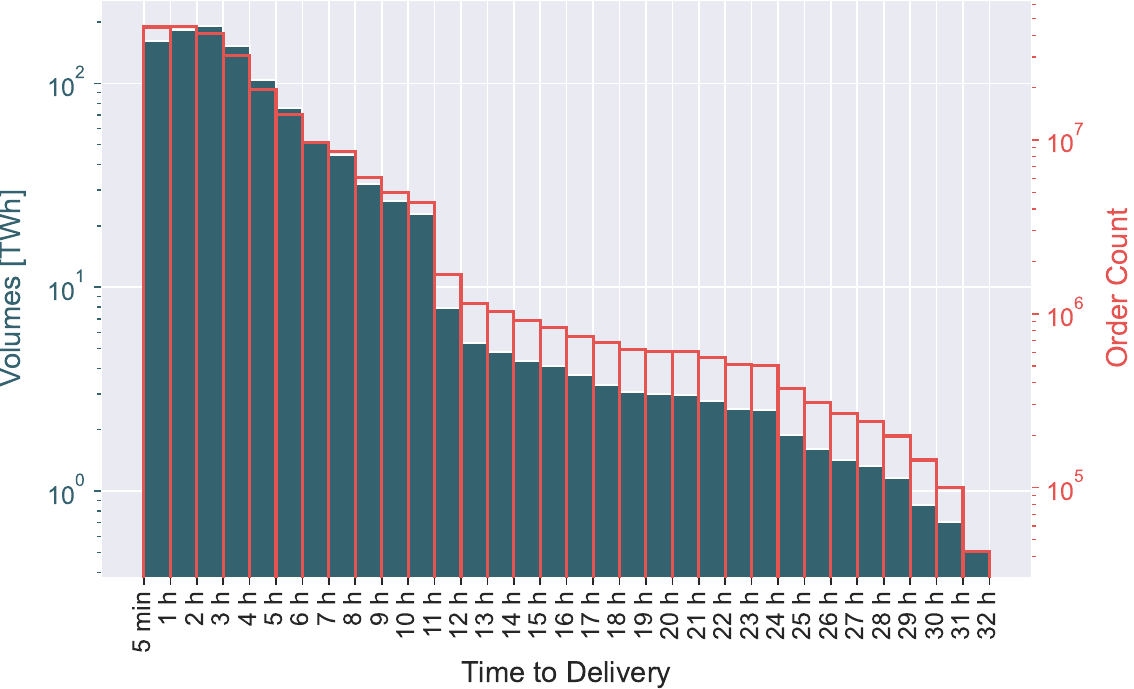}
    
    \caption{\label{fig:liquidity}The binned logarithmic distribution of time intervals between order submission and physical delivery for the German continuous intraday market in 2021. The teal bars illustrate the volume traded, while the red bars show the same data, counting the number of orders submitted. Both distributions exhibit an approximately exponential increase as the time to delivery shortens. Notably, there is a distinct gap between the 10-hour and 11-hour bins, where a significant increase in trading activity is observed.}
\end{figure}

A key challenge of trading in the continuous ID market is the lack of liquidity several hours before delivery. This results in wide bid-ask spreads and high price volatility. An example of this liquidity disparity is shown in Figure \ref{fig:liquidity}, which reveals that while a few hours before delivery typically thousands of orders are posted and cleared for a given product, this number drops to several hundred for products whose delivery is $10$ hours or longer away and to a handful of orders at the start of trading.

As explained in more detail in \citep{kuppelwieser2023intraday}, low liquidity far from delivery occurs because market participants lack precise information regarding their expected production and demand. As the delivery time approaches, weather forecasts and demand predictions become significantly more accurate. This increased certainty allows participants to take more precise balancing actions, which consequently drives up trading activity and liquidity.

Low liquidity creates difficulties for automated trading strategies that aim to capitalize on price differences throughout the day. For instance, evening products traded during the morning hours often fail to reflect prevailing market prices accurately, as they are largely influenced by market makers rather than real supply and demand dynamics. This issue of liquidity is the main disadvantage of the continuous market, compared to its auction alternatives. It reduces the efficiency of price discovery and also increases traders' risk, as they may experience suboptimal pricing and reduced profitability when locking in positions early, as discussed more in-depth in e.g. \citep{graf_frequent_2018}.

Finally, Table \ref{tab:epex_spot_volume} provides a concise summary of key statistics for the Central-Western European (CWE) Spot Market over recent years. The continuous market stands out as the dominant intraday trading platform. Specifically, the German continuous intraday market in 2021, which serves as the data source for this study, recorded 240.4 million order submissions (including order changes). This translates to an average of 7.6 orders placed per second. The sheer volume and frequency of transactions underscore the necessity of accurately modeling and responding to high-frequency market dynamics, to realistically simulate any trading strategy on this market. This calls for extremely fast and reactive trading. In the next section, we will propose an algorithm that can easily keep pace with the market's speed.

\begin{table}[t]
\centering
\resizebox{\textwidth}{!}{%
\begin{tabular}{lrrrrrrr}
\toprule
\textbf{} & \textbf{2023} & \textbf{2022} & \textbf{2021} & \textbf{2020} & \textbf{2019} & \textbf{2018} & \textbf{2017} \\
\midrule
\textbf{Day-Ahead}       & (+16.9\%) 419.3 & (-8.6\%) 358.8 & (-4.5\%) 392.7 & (-8.4\%) 411.1 & (+3.1\%) 448.8 & (+5.3\%) 435.2 & 413.2 \\
\textbf{ID-Auction}      & (+9.0\%) 9.41   & (-1.9\%) 8.63  & (+0.0\%) 8.8    & (+20.5\%) 8.8   & (+7.4\%) 7.3   & (+30.8\%) 6.8  &  5.2   \\
\textbf{ID-Continuous}   & (+29.9\%) 119.8 & (+8.3\%) 92.2   & (+10.1\%) 85.1  & (+23.5\%) 77.3  & (+11.4\%) 62.6 & (+11.3\%) 56.2 & 50.5  \\
\bottomrule
\end{tabular}
}
\caption{Combined EPEX Spot trading volumes in TWh for the CWE region (AT, BE, DE/LU, FR, NL) taken from the EPEX Spot annual reports \citep{epex_annual_reports}. The percentage-change to the respective last year is given in brackets. Clearly, intraday trading is gaining relevance compared to DA auctions, while the ID auction market is slowly picking up in activity.}
\label{tab:epex_spot_volume}
\end{table}

\section{An Efficient Rolling Intrinsic Policy for Continuous Markets} \label{sec:method}
In this section, we describe the intrinsic problem for continuous markets, discuss the rolling intrinsic policy as an improvement of the intrinsic strategy, and present an efficient implementation as a sequence of dynamic programming problems.

\subsection{The Intrinsic Problem} \label{ssec:intrinsic}
The intrinsic strategy is a simple policy to maximize the rewards of a storage system, which is widely used in the literature on gas storage and in the industry to optimize storage assets \citep[e.g.][]{gray2004towards, GrKh04b, Lai10, Lai11, lohndorf2021gas}. In the extant literature, the intrinsic problem is set in a perfectly liquid futures market, trading $T \in \Nd$ futures of a commodity that can be stored and which trade for prices $P_t$. To not overcomplicate notation, we assume here that futures contracts are for delivery in consecutive time periods $[0,1), [1, 2), \dots, [T-1, T)$ and that there are no overlapping contracts. To fix ideas, one can imagine a battery storage system trading futures for delivery in the $24$ hours of the next day.
    
The storage operator decides about traded quantities $q_t$ (MWh) of these contracts such that volumes bought are transferred into the storage, while sold volumes are physically fulfilled from the storage. In particular, in the simplest case, the intrinsic optimization boils down to the following linear problem of optimizing the storage level $(s_t)_{t=1}^T$ (MWh) and the future positions $(f_t)_{t=1}^T$ (MWh) in the following way
\begin{align} \label{eq:simple_intrinsic}
    \begin{array}{lll}
        \underset{q_t}{\max} & \sum\limits_{t=1}^T -P_tq_t \\
        \text{s.t.} & f_{t} = f^0_{t} + q_t, \qquad & \forall t = 1, \dots, T \\
                    & f_t \in [\underline f, \bar f], & \forall t = 1, \dots, T \\
                    & s_{t} = s_{t-1} + f_t, & \forall t = 1, \dots, T \\
                    & s_{t} \in [0, \bar s], & \forall t = 1, \dots, T,
    \end{array}
\end{align} 
where $s_0$ and $(f_t^0)_{t=1}^T$ are the initial storage level and the initial future positions, respectively, which are data to the problem. In the above problem, positive $q_t$ represent buying decisions while negative ones represent sells. Furthermore, $\bar s$, $\underline f$, and $\bar f$ are physical bounds for the maximum storage level (MWh) as well as for withdrawal and injection (MWh). The respective constraints ensure that the position is \emph{physical} in the sense that traded quantities can be accommodated in the storage without violation of the limits on injection and withdrawal. 

The above linear problem is the simplest form of the intrinsic, which assumes a perfectly efficient storage, liquid markets, and no cost for operating the storage. Furthermore, the problem is myopic in the sense that the decision maker uses all the flexibility of the storage instantly without anticipating future changes in prices, which would make it more advantageous to wait before trading immediately. The myopic nature of the intrinsic policy makes it suboptimal and motivates its name originating from the \emph{intrinsic value} of a financial option contract, i.e., the value an option has when it is immediately exercised, not taking into account the value of waiting (the extrinsic value). Problem \eqref{eq:simple_intrinsic} is the version of the intrinsic problem prevailing in the literature.

Next we will extend \eqref{eq:simple_intrinsic} to a policy for a BESS trading on continuous intraday power markets, explicitly taking into account order book information, trading costs, and a simple linear model for efficiency and degradation losses of batteries. In order to do so, we reuse equation \ref{eq:simple_intrinsic}'s notation and introduce a set of futures contracts $\Tc_{t^*}$ at time $t^*$, for delivery of electricity in non-overlapping time periods in the future. The time when the earliest contract goes into delivery is denoted by $t_0$. Note that the formulation could easily be extended to futures contracts that overlap, such as quarter-hourly and hourly contracts. For simplicity, in the following we only focus on the intrinsic at a singular time $t^*$, denoting $\Tc_{t^*}$ simply as $\Tc$.

For each contract $t \in \Tc$, there is an order book $\Oc_t = \Oc_t^+ \cup \Oc_t^-$ with $\Oc_t^+$ containing all the currently active asks and $\Oc_t^-$ containing all the currently active bids for contract $t$. For each of these orders $i \in \Oc_t$, we define a limit price $P_i$ and a quantity $Q_i>0$. Within this framework, the optimization evaluates the exact historical limit order book to determine which orders to match. By design, the objective function naturally prioritizes the most competitive prices, always taking the lowest available asks and highest available bids, since selecting a less favorable price for a given volume would yield a mathematically suboptimal result. Furthermore, we introduce efficiency factors $\eta^+$, $\eta^- \in (0,1]$ for charging and discharging, respectively, as well as a volume-based trading cost $\nu_{\text{trade}}$ for matching orders on the market. 

Batteries have a limited life. Roughly speaking, there are two effects causing battery degradation: calendar age, i.e., the loss of capacity with time even without active use of the battery, and cycling, i.e., the wear and tear caused by injection and withdrawal of energy. Since we are looking at short-term planning, we focus on the latter and model a symmetric linear cycling (i.e. degradation) cost $\nu_\text{deg}$ in \texteuro/MWh of injected/withdrawn energy. The intuition behind this choice is that cells will have to be replaced after a certain number of cycles, which induces a replacement cost. With this information, we calculate a proportional per cycle and ultimately a per MWh cost for using the battery. We note that actual degradation depends not only on the amount of cycled energy but also on the depth of discharge, temperature, and several other factors \citep[see e.g.][]{xu_factoring_2017, graef_2022}{}. However, for the sake of simplicity, we focus on the amount of cycled energy, which is the most important factor. Combining the cost of trading and the cost of degradation, we define $\nu = \nu_{\text{trade}} + \nu_\text{deg}$ as the combined variable cost in units of \texteuro/MWh. By acting as a friction term in the objective function, this combined cost ensures that the battery only executes trades where the market price spread exceeds the cost of physical wear and tear. Consequently, the degradation model actively prevents marginal trades, which reduces the total executed trading volume and theoretical profits, but yields a realistic and economically viable dispatch schedule.

Until this point, all discussed extensions of \eqref{eq:simple_intrinsic} can still be modeled as a linear program (LP). However, there are two aspects of intraday trading that introduce non-linearities. Firstly, trade quantities on electricity exchanges are not continuous variables but multiples of the minimum trading unit $u$, which implies that they are of the form $q_i = k_i u$ with $k_i \in \Nd_0$. Secondly, efficiencies smaller than $1$ coupled with negative prices make it profitable to spend energy by simultaneously charging and discharging, which is physically not possible for a battery. More specifically, for a single product $t$ and orders $i \in \Oc^+_t$, $j \in \Oc^-_t$, buying a quantity $q$ using order $i$ and immediately re-selling the power $\eta^- (\eta^+ q)$ that remains after accounting for efficiency losses using order $j$ yields a positive profit if
\begin{equation}
    -q P_i + \eta^- (\eta^+ q)  P_j > 0 \Leftrightarrow  (\eta^- \eta^+) P_j > P_i.
\end{equation}

Since $P_i > P_j$ by definition, this can only happen if both prices are negative. However, because it is physically not possible to charge and discharge the battery at the same time, such trades have to be prevented.

With these preparations, we can write the modified intrinsic as the following mixed integer linear problem (MILP):
\begin{align}
\label{eq:intrinsic_vanilla}
       \underset{f_t, s_t, \alpha_t, q_i, k_i}{\max} \quad & \sum\limits_{t \in \Tc} \left( \sum\limits_{i \in \Oc_t^-} (P_i - \nu) q_i - \sum\limits_{i \in \Oc_t^+} (P_i+\nu) q_i \right) \\
            \text{s.t.} \quad	& 0 \leq q_i \leq Q_i, & \forall i \in \Oc_t, \forall t \in \Tc \notag\\
                                    & q_i = k_i u, & \forall i \in \Oc_t, \forall t \in \Tc \notag\\
                                    & f_t^+ = \sum_{i\in \Oc_t^+} q_i, & \forall i \in \Oc_t^+, \forall t \in \Tc \notag\\
                                    & f_t^- = \sum_{i\in \Oc_t^-} q_i, & \forall i \in \Oc_t^-, \forall t \in \Tc \notag\\
                                    & f_{t} = f^0_{t} + f_t^+ - f_t^-, \qquad & \forall t \in \Tc \notag\\
                                & f_t \in [\underline f, \bar f], & \forall t \in \Tc \notag\\
                                    & f_t = i_t - w_t, & \forall t \in \Tc \notag\\
                                    & i_t \in [0, \alpha_t \bar f], & \forall t \in \Tc \notag\\
                                    & w_t \in [0, -(1-\alpha_t) \underline f], & \forall t \in \Tc \notag\\
                                & s_{t} = s_{t-1} + \eta^+ i_t - \frac{1}{\eta^-} w_t, & \forall t \in \Tc \notag\\
                                & s_{t} \in [0, \bar s], & \forall t \in \Tc \notag \\
                                    & \alpha_t \in \{0,1\}, \; k_i \in \Nd , \; \nu \in \mathbb{R}_{\geq 0}& \forall i \in \Oc_t, \forall t \in \Tc \notag
\end{align}
where the $\alpha_t$ are binary variables which are equal to $1$ in case the decision for contract $t$ is to buy and $0$ if the decision is to sell, and $i_t$ ($f_t^+$) and $w_t$ ($f_t^-$) are the combined buy and sell positions per product including (excluding) previous positions.

Note that, due to the above discussion, strictly speaking, the binary variables $\alpha_t$ are only required for hours $t$ where both $\Oc_t^+$ and $\Oc_t^-$ contain orders with negative prices. Since $u$ is usually rather small and negative prices do not occur too often, typically linear relaxations of the above problems are good and yield near-optimal solutions in the root node of the branch-and-bound trees, making the problem a comparably easy MILP in most realistic cases.

\subsection{The Rolling Intrinsic}
\begin{algorithm}[t]
    \begin{algorithmic}[]
        \State \textbf{Input}: Initial storage $s_0$, start-time $t_{start}$ and stop-time of simulation $t_{end}$, limit order book information $\Tc_{t^*} \ \forall t^* \in [t_{start}, t_{end}]$, solve-frequency of intrinsic, other parameters $(\bar f,\underline f, \nu, \dots)$
        
        \State \textbf{Output}: Final battery schedule during simulation runtime, cumulated trade profits
        
        \For{time instant $t^* \in [t_{start}, t_{end}]$}
            \State Solve the intrinsic according to equation \eqref{eq:intrinsic_vanilla}

            \State Update all positions $f_t, \forall t \in \Tc_{t^*}$

            \State Log profits/losses from trading
        \EndFor
    \end{algorithmic}
    \caption{\label{alg:RI} The rolling intrinsic}
\end{algorithm}

The rolling intrinsic policy (RI), as outlined in Algorithm \ref{alg:RI} was originally introduced in \cite{gray2004towards} and is a dynamic extension of the intrinsic value discussed in the last section. Starting from an initial storage state $s_0$ and market positions $(f_t)_{t \in \Tc}$ acquired in previous periods, the RI repeatedly checks for chances of profitable rebalancing by re-running the intrinsic policy. Although the resulting decisions are still myopic, the RI clearly represents an improvement over the static intrinsic policy that does not adapt positions at all. Furthermore, the myopic nature of the decisions has the advantage that the RI does not speculate but only enters immediately profitable positions and therefore does not run the risk of accumulating losses. For these reasons and because of its conceptual simplicity and relatively low computational cost, the RI has gained widespread industry adoption.

While the policy remains strictly myopic, its computational simplicity is precisely what allows it to be evaluated at extremely high frequencies. In a continuous market, the ability to instantly capture fleeting, short-lived price spreads through rapid execution often yields an operationally stronger and more profitable strategy than theoretically more sophisticated but computationally slow and practically unimplementable models.

\subsection{A Dynamic Programming Formulation of the Rolling Intrinsic Policy} \label{ssec:intrinsic_dp}
In this section, we describe how to reformulate problem \eqref{eq:intrinsic_vanilla} to a dynamic programming formulation that can be solved by an efficient implementation of the dynamic programming algorithm. The aim of the reformulation is to make faster decisions to keep up with the pace of the intraday market, where typically many orders arrive per second, requiring constant rapid reevaluation of one's position.

The main idea of the reformulation is to assign each product $t \in \Tc$ a stage in which the decisions to accept orders $i \in \Oc_t$ are taken while the storage level is the state variable that connects the stages. Hence, in the language of dynamic programming, the decisions $q_i$, $k_i$, $\alpha_t$, and $f_t$ are the actions, and the resulting storage state $s_{t}$ is the argument of the value function of the next stage. The solution of the intrinsic itself remains a single-step optimization, executed at a specific point in time. To solve this high-dimensional decision problem efficiently, we break up the logic of the underlying mixed-integer linear program into individual products. Mathematically, each delivery hour is mapped to a distinct "stage" of the DP. Note that this construction is artificial, as the DP stages do not represent different points in time when sequential decisions are made, but merely different parts of a single trading decision that is taken all at once.

We assume without loss of generality that $\Tc_t = \{1, \dots, T\}$ set $V_{T+1} \equiv 0$ and for $t = 1, \dots, T$ define the value function as
\begin{align}
    V_t(s_{t-1}) = \left\{  \begin{array}{lll} 
       \underset{f_t, s_t, \alpha_t, q_i, k_i}{\max} \quad & \multicolumn{2}{l}{\sum\limits_{i \in \Oc_t^-} (P_i - \nu) q_i - \sum\limits_{i \in \Oc_t^+} (P_i + \nu) q_i + V_{t+1}(s_t)}\\
            \text{s.t.} \quad	& 0 \leq q_i \leq Q_i, & \forall i \in \Oc_t \\
                                & q_i = k_i u, & \forall i \in \Oc_t \\
                                & f_t^+ = \sum_{i\in \Oc_t^+} q_i, & \forall i \in \Oc_t^+ \\
                                & f_t^- = \sum_{i\in \Oc_t^-} q_i, & \forall i \in \Oc_t^- \\
                                & f_{t} = f^0_{t} + f_t^+ - f_t^- \label{eq:intrinsic_dp_long}\\
                                & f_t \in [\underline f, \bar f]\\
                                & i_t \in [0, \alpha_t \bar f] &  \\
                                & w_t \in [0, -(1-\alpha_t) \bar f] &  \\
                                & s_{t} = s_{t-1} + \eta^+ i_t - \frac{1}{\eta^-} w_t & \\
                                & s_{t} \in [0, \bar s]\\
                                & \alpha_t \in \{0,1\}, \; k_t \in \Nd, \nu \in \mathbb{R}_{\geq 0}.&
    \end{array} 
    \right.
\end{align}
This formulation might not seem very natural, since in reality the decisions in all stages of the above problem happen simultaneously instead of sequentially as is usually the case in dynamic programs. Because of this, existing studies and industry practice almost exclusively rely on monolithic MILP or LP formulations, which offer the most intuitive and direct mathematical mapping of the battery's physical constraints. However, reformulating this simultaneous problem into a multi-stage dynamic program is the key to unlocking significantly faster computation times. As we will demonstrate in Section \ref{sec:results}, this novel perspective allows us to solve the intrinsic optimization up to three orders of magnitude faster than a standard MILP, executing several solves per millisecond. This speedup is not merely a convenience; it is a strict prerequisite for realistically tracking and reacting to the millisecond-pace of the continuous limit order book.

We propose solving the problem by an application of the dynamic programming algorithm, discretizing the actions $f_t$ into steps $\kappa \cdot u$ ($\kappa \in \Nd$) apart, where $u$ is the minimal tradable quantity. Furthermore, we introduce a function $\pi_t$ that encodes the current order book information and returns the cost or revenue of buying or selling on the intraday market, taking into account trading costs. Furthermore, we define a state transition function 
\begin{equation}
    S(s_{t-1}, f_t) = \begin{cases} s_{t-1} + \eta^+ f_t, & \text{if } f_t > 0 \\ s_{t-1} + \frac{f_t}{\eta^-}, & \text{otherwise.} \end{cases}
\end{equation}

With these preparations, we can now re-write problem \eqref{eq:intrinsic_dp_long} as
\begin{equation} \label{eq:intrinsic_dp}
        V_t(s_{t-1}) = \left\{ \begin{array}{ll} 
           \underset{k_t}{\max} \quad & \multicolumn{1}{l}{\pi_t(f_t^0 - f_t) + V_{t+1}(S(s_{t-1}, f_t))}\\
            \text{s.t.} \quad & f_t = f_t^0 + k_t u \\
                              & k_t \in \left[-\frac{\eta^-s_{t-1} + f_t^0}{u}, \frac{\bar s - s_{t-1}}{\eta^+ u} - \frac{f_t^0}{u}\right] \cap \left[\frac{\underline{f} - f_t^0}{u}, \frac{\bar f - f_t^0}{u}\right] \cap  \kappa\Zd,
        \end{array} 
        \right.
\end{equation}
where the bounds in the first constraint $k_t$ enforce the storage energy limits, while the second interval enforces the power limits. Given $V_{t+1}$, solving \eqref{eq:intrinsic_dp} boils down to evaluating the objective for all feasible $k_t$, i.e., to a small number of arithmetic operations.

Note that the above reformulation of \eqref{eq:intrinsic_dp_long} to \eqref{eq:intrinsic_dp} automatically takes care of the issue of simultaneously accepting buy and sell orders in the case of negative prices, and, by fixing the trading step-size $u$ to a multiple of the minimum bid-size, naturally yields feasible position sizes.

To solve \eqref{eq:intrinsic_dp}, we need to know the value functions $V_{t+1}$. In line with the usual dynamic programming algorithm, we calculate $V_t$ going backward in time: $V_{T+1}$ is the known boundary condition, which makes it possible to calculate $V_T$ in any given state. However, due to efficiencies $\eta^\pm < 1$, the trades that live on a finite grid translate to storage levels that do not sit on a regular grid, especially after several rounds of injections and withdrawals. 

Therefore, unlike classic applications of the dynamic programming algorithm, we do not define $V_T$ to have a finite domain, but instead on the whole interval of possible storage states $[0, \bar s]$. In order to obtain the value function, we discretize the possible storage states $[0, \bar s]$ to a finite grid $G \subset [0, \bar s]$ with $\lvert G \rvert = m$. We then approximate the value function $V_{T}$ by evaluating it on the grid and linearly interpolating for points $s \notin G$. Thus, the parameter $m$ represents the number of discrete grid points we use to model the possible state of charge of the battery within our dynamic programming approach. A smaller storage discretization, such as $m=11$, means we evaluate the value function at fewer points, resulting in faster computation times but a slightly less accurate approximation of the optimal strategy. 

In particular, we denote the resulting approximations by $\tilde V_{t+1}$ and describe our linear approximation between the neighboring discretized gridpoints $s_i, s_{i+1} \in G$ as follows:
\begin{equation}
    \tilde{V}_{t+1}(s_t) = \frac{m(s_{i+1} - s_t)}{\bar s}V_{t+1}(s_i) + \frac{m(s_t - s_i)}{\bar s}V_{t+1}(s_{i+1}),
\end{equation}
for $s_{i} \leq s_t \leq s_{i+1}$. More sophisticated interpolations, taking into account, for example, the curvature of the value function (given by the buy and sell price stacks), could potentially enhance the precision of our optimization. However, since such approaches would be computationally more costly and the obvious choices, such as quadratic interpolations or cubic splines, did not show any improvements in our numerical experiments, we remain with the linear approximation above.

Once we have obtained $\tilde V_T$, we can solve the problem that defines $V_{T-1}$ with $V_T$ replaced by $\tilde V_{T}$. We then repeat the process until we have obtained the approximations $\tilde V_t$ of all the value functions. Note that, in principle, $G$ can be chosen independently for every $t$ and can also change between various solutions of the intrinsic problem. For the sake of simplicity, we keep the grid constant, i.e. $G_{t^*,t} \equiv G$. Thus our DP's computational complexity grows linearly as $\mathcal{O}(\frac{|T|\cdot |G|}{\kappa})$.

Algorithm \ref{alg:DP_RI_algo} displays an overview of one such dynamic program solution run of the intrinsic problem at time $t^*$, where $\Xc_t$ is the set of all possible actions $k_t$ as defined in equation \eqref{eq:intrinsic_dp}.
\begin{algorithm}[t]
    \caption{DP at runtime $t^*$}
    \label{alg:DP_RI_algo}
    \begin{algorithmic}[]
        \State \textbf{Input}: $f_t^0$,  $\Oc_t$ for all tradable products $t \in \Tc$, other parameters $(\kappa, s_0, \nu, G, \dots)$
        \State \textbf{Output}: Updated market positions $f_t$
        \State Set $V_{T+1} \equiv 0$.
        \Statex
        \State \textbf{Backwards Pass}
        \For{$t=T, \dots, t_0$}
            \For{each state $s \in G$}
                \State $V_t(s) \gets \max_{k_t \in \Xc_t(s, f_t^0)} \set{ \pi_t(k_tu) + \tilde V_{t+1}(S(s, k_tu))}$
            \EndFor
        \EndFor
        \Statex
        \State \textbf{Forward Pass}
        \For{$t=t_0, \dots, T$}
            \State $\hat{f_t}\gets \argmax_{k_t \in \Xc_t(s, f_t^0)} \set{\pi_t(k_tu) + \tilde V_{t+1}(S(s_{t-1}, k_tu))}$
            \State $f_t^0 \gets \hat{f_t}$
            \State $s_t \gets S(s_{t-1}, \hat{f_t})$
        \EndFor
    \end{algorithmic}
\end{algorithm}

The approximation of the value functions $\tilde V_t$ introduces an error relative to the exact MILP formulation of the problem. In the following proposition, we show that the domain of the true functions $V_t$ is not actually continuous and that consequently there is a grid $G$, which makes the approximation and thereby the DP formulation of the problem exact.
\begin{prop}
    For all $t \in \set{1, \dots, T}$, there is a finite grid $G_t$ that contains all possible storage states at the beginning of period $t$. Choosing $G$ to approximate the value functions alleviates the need for interpolation and makes the approximation exact.
\end{prop}
\begin{proof}
    Suppose that the initial state of charge (SoC) is $s_0$, we set $\kappa=1$ and look at the simplest case of $f_t^0 = 0$ for all $t \in \Tc$. The possible storage states at the beginning of stage $t_0+1$ are
    \begin{equation}
        G_{t_0+1} = \set{s_0 + uk_1^+\eta^+ - \frac{uk_1^-}{\eta^-}: k_1^\pm\in \Kc_1^\pm},
    \end{equation}
    and $\Kc_1^+=\set{ k\in \Nd_0: k\leq u^{-1} \min \left(\frac{\bar s-s_0}{\eta^+}, \bar f \right)}$ and $\Kc_1^-=\set{ k\in \Nd_0: k\leq u^{-1}\min \left( s_0\eta^-, -\underline f \right)}$.

    Defining the sets $G_t$ and $\Kc_t^\pm$ going forward in time, for period $t > t_0$ we thus have, given that $\Kc_{t-1}^{\pm}$ is already known, 
    \begin{equation}
        G_t = \set{s_t=s_0 + uk_t^+\eta^+ - \frac{uk_t^-}{\eta^-}: k_t^\pm\in \Kc_t^\pm, 0\leq s_t \leq \bar s}
    \end{equation}
    with 
    \begin{align}
        \Kc_t^+&=\set{k \in \Nd_0: k \leq \frac{(\bar s - s_0) + u\max (\Kc_{t-1}^-)(\eta^-)^{-1}}{u \eta^+}} \\
        \Kc_t^-&=\set{k \in \Nd_0: k \leq u^{-1} \eta^-\left( s_0+\eta^+u\max (\Kc_{t-1}^+)\right)}.
    \end{align}
    Note that the way $G_t$ and $\Kc_t^\pm$ are defined takes care of repeated charging and discharging in periods $t_0, \dots, t$ and therefore covers all possible storage states.
\end{proof}
Choosing the finite grids described above would ensure that the DP solution is exact and equivalent to the MILP formulation, as the value function only needs to be evaluated at points in the grid and there is no need for approximation. This comes at the cost of computational efficiency, as choosing $G$ to be exact, would yield a very fine grid for cases where there are many tradable products $\Tc$ and $\eta^\pm < 1$.
The formulation presented above adds to the extant literature in two important ways:
\begin{enumerate}
    \item \textbf{Full LOB Granularity:} Unlike standard rolling intrinsic implementations that rely on aggregated price curves, transaction data, or simplified mid-market prices, our approach explicitly models the \textit{full limit order book dynamics}. This allows the policy to capture liquidity constraints and depth-dependent pricing exact to the historic microstructure level.
    
    \item \textbf{High-Frequency Feasibility via DP:} The reformulation from a generic MILP to a specialized dynamic programming scheme reduces the computational burden by orders of magnitude. This speedup is not merely an efficiency gain, it is a structural prerequisite that enables the backtesting of a strategy that re-optimizes on a millisecond timescale over a full year of market tick-data, a task that would be computationally intractable with standard solvers.
\end{enumerate}

\subsection{From Optimization to Execution}

The optimal position changes $\hat{f_t}$ derived from the forward pass must be translated into actionable limit orders. To mitigate execution risk in a fast-moving market, we submit these quantities as fill-or-kill orders. This ensures that the strategy only trades a position if the full desired volume is available at the expected price, preventing partial fills that could lead to suboptimal inventory states. However, due to the latency $\Delta t$ between the optimization start and order arrival at the exchange, market conditions may shift, potentially leaving orders unmatched or resulting in positions that, while optimal at $t^*$, are no longer physically feasible given the market's state at $t^* + \Delta t$. In such cases, the rolling nature of the policy corrects these deviations in the subsequent optimization step, naturally steering the system back toward feasible, profitable operation.

Furthermore, it is important to emphasize that while our algorithmic trading strategy operates at a high financial frequency, this does not strictly equate to high-frequency physical cycling of the battery. Frequent re-evaluations of the market state primarily allow the algorithm to adjust the timing of financial positions as the order book evolves. However, re-optimizing a myopic schedule in a volatile continuous intraday (CID) market can influence the physical dispatch, and thus the wear-and-tear of the battery, in three primary ways:
            
\begin{enumerate}
    \item \textbf{Shifting volume:} The algorithm might shift the traded electricity between hours (e.g., moving delivery to a neighboring hour) to capture a better price. This keeps the total physical volume constant, though it requires the original positions to be unwound.
    \item \textbf{Increasing volume:} The optimization may encounter new, independent price spreads, not affecting previous positions. Capitalizing on these new spreads layers additional charge and discharge actions onto the existing schedule, strictly increasing the battery's physical cycles.
    \item \textbf{Reducing volume:} If market prices move such that an originally profitable spread collapses or inverts, the optimal action may be to only unwind the existing positions, thereby reducing the scheduled physical cycles.
\end{enumerate}

In practice, any combination of these effects can occur for a given optimization instance. Depending on market conditions, total physical cycles might therefore increase or decrease as a result of high-frequency re-optimization. Consequently, our combined linear degradation and trading costs naturally act as necessary friction terms against weak market signals. This prevents the algorithm from constantly chasing marginal price movements, thereby regulating the schedule revisions and limiting the battery's actual physical wear and tear.

By contrast,  alternative reduced-order market formulations, such as volume-weighted average prices, are computationally lighter and allow for richer battery models, as they frequently overestimate revenues by ignoring strict limit order book liquidity constraints. By operating directly on the limit order book, our framework prevents the policy from attempting to clear energy volumes that do not physically exist at the expected price. Furthermore, our strategy provides a clear, directly implementable policy on how and when the battery's orders should be submitted to the market and how one must dynamically correct missed positions to avoid physical imbalances.




\section{\label{sec:results}Results}
In this section, we present the results of a numerical case study. In Section \ref{ssec:case_study_setting}, we discuss the setting and implementation of the case study. In Section \ref{ssec:dp_vs_milp}, we measure the precision of our DP solution and compare the runtimes with the exact MILP formulation, arguing that the DP leads to runtimes that are orders of magnitude faster. In Section \ref{ssec:yearly_comp}, we present the results of our policy for a whole year of trading and measure the impact of trading speed on performance, while in Section \ref{ssec:param_eval}, we examine how changes in storage and optimization parameters affect the results.

\subsection{Setting and Mode of Comparison} \label{ssec:case_study_setting}
Unless otherwise stated, for our case study, we use a BESS with an energy-to-power ratio (duration) of 1 hour, set $\bar f = - \underline f = 10$ MWh (i.e., a power of 10 MW) and $\bar s = 10$ MWh, and assume efficiency losses of $\eta^+ = \eta^- = 0.95$, i.e., a round trip efficiency of $\eta^+\eta^- \approx 0.9$ similar to \cite{cole2023battery, bess_revenue_index_1h}. 

We use linear degradation costs of $\nu_\text{deg} = 4$ \texteuro/MWh, i.e., a round-trip degradation cost of $8$ \texteuro/MWh for every MWh of electricity sold by the battery. We motivate this value by \citet{goldmansachs_2023} who project the undiscounted cost of battery replacement in 2030 to be around $60$k \texteuro / MWh, a 10-year battery warranty, and approximately $2$ charging cycles per day. Dividing the replacement costs by cycles gives the cost of one full cycle, which allows us to infer the induced \texteuro /MWh linear degradation cost: $8 \approx 60k / (2 \cdot 365 \cdot 10)$.

Furthermore, we choose $\nu_\text{trade} = 0.09$ \texteuro/MWh, reflecting the current EPEX volume-based trading fees per each matched order. In total, the variable cost $\nu = \nu_{\text{trade}} + \nu_\text{deg}$ thus takes the value of $\nu = 4.09$ \texteuro/MWh. This variable, symmetrically applied, cost $\nu$ is always present in the optimization, while $\nu_\text{trade}$ is deducted from profits after each transaction. The degradation cost $\nu_\text{deg}$ is subtracted only once, after each product's final schedule is set at gate closure.

All results are generated using historical 2021 order book data from the EPEX Spot continuous intraday market for Germany, focusing on hourly products. It is not possible ex post to simulate iceberg orders correctly given the EPEX dataset, so we omit them in our simulations. The exchange allows for price increments of $0.01$  \texteuro  \ and volume increments at $0.1$ MWh, defining the minimum action size on the market as $u = 0.1$ MWh. The discretization of the dynamic programming action is defined at the most accurate minimum value of $\kappa = 1$ throughout.

Clearly, in real-life trading, there will always be some delay between two consecutive optimizations. We explicitly model this delay $\Delta t$ and conceptually split it into two components: the \textit{solve-time} and the \textit{technical delay}:

\begin{figure}[h!]
    \centering
    \begin{tikzpicture}
\draw[thick, -Triangle] (0,0) -- (6,0);

\draw (0 cm,3pt) -- (0 cm,-3pt);
\draw (2 cm,3pt) -- (2 cm,-3pt);
\draw (5 cm,3pt) -- (5 cm,-3pt);

\node[font=\scriptsize, text height=1.75ex,
text depth=.5ex] at (0,-.3) {$t^*$};
\node[font=\scriptsize, text height=1.75ex,
text depth=.5ex] at (5,-.3) {$t^*+\Delta t$};


\draw[customRed, line width=4pt] (0,.5) -- +(2,0);
\draw [thick ,decorate,decoration={brace,amplitude=5pt}] (0,0.7)  -- +(2,0) 
       node [black,midway,above=4pt, font=\scriptsize] {Solve Time};

\draw[customGreen, line width=4pt] (2,.5) -- +(3,0);
\draw [thick ,decorate,decoration={brace,amplitude=5pt}] (2,0.7)  -- +(3,0) 
       node [black,midway,above=4pt, font=\scriptsize] {Technical Delay};



\draw [thick,decorate,decoration={brace,amplitude=5pt}] (5,-.7) -- +(-5,0)
       node [black,midway,font=\scriptsize, below=4pt] {Waiting Duration};

\end{tikzpicture}
    \label{fig:delays}
\end{figure}
\vspace{-1em} 
The solve time represents the time required to find an optimal solution for a given intrinsic problem. The technical delay accounts for various operational latencies, such as communication delays with the exchange, delays caused by internal database and trading systems, and other post-solve delays before the submitted orders are matched on the exchange.
If the combined delay exceeds $1$ ms (the tick duration of the EPEX Spot market), this has two effects on trading: Firstly, there is the risk that an order placed by our policy may not be matched since its intended counterpart either expired or was cleared against other market participants' orders in the meantime. We deal with this by submitting all-or-none orders that expire if they are not immediately cleared, which might lead to non-physical positions, i.e., market positions that cannot be matched by a feasible schedule of injections and withdrawals. In such a situation, further trades are required to correct the position so that the final battery schedule becomes physically feasible. However, in the way the intrinsic problem is set up, such positions would be corrected in the next call of the intrinsic problem, forcing the RI to return to a physical position as quickly as possible.

The second issue arising from delays is that the policy might not be able to react to every single order book update. In particular, we assume that the next intrinsic is solved only at $t^*+\Delta t$ observing the updated order book state at this time. Hence, if there are updates to the order book in the time interval $(t^*, t^*+\Delta t)$, the policy will not be able to react immediately to it.
We give a short overview of our chosen standard parameters in Table \ref{tab:baseline_params}, and provide a graphic overview of the simplified execution logic in Figure \ref{fig:flowchart}.
\begin{table}[h!]
    \centering
    \footnotesize
    \begin{tabular}{lccc}
        \toprule
        \textbf{Parameter} & \textbf{Symbol} & \textbf{Value} & \textbf{Unit} \\
        \midrule
        \multicolumn{4}{l}{\textit{Battery Storage Specs}} \\
        Energy Capacity & $\bar{s}$ & 10 & MWh \\
        Battery Duration & -- & 1 & h \\
        Efficiency (Charge/Discharge) & $\eta^+, \eta^-$ & 0.95 & -- \\
        \midrule
        \multicolumn{4}{l}{\textit{Costs \& CID Market}} \\
        Trading Fee & $\nu_{\text{trade}}$ & 0.09 & \texteuro/MWh \\
        Degradation Cost (Linear) & $\nu_{\text{deg}}$ & 4.00 & \texteuro/MWh \\
        Min. Tradable Unit & $u$ & 0.1 & MWh \\
        Simulated Year & -- & 2021 & -- \\
        \midrule
        \multicolumn{4}{l}{\textit{Optimization Settings}} \\
        Action Discretization & $\kappa$ & 1 & -- \\
        Storage Discretization Size & $|G|$ & 11 & -- \\
        \bottomrule
    \end{tabular}
    \caption{Baseline simulation parameters for the BESS, the intraday market environment and the DP.}
    \label{tab:baseline_params}
\end{table}

\begin{figure}[h!]
    \centering
    \resizebox{0.6\textwidth}{!}{\begin{tikzpicture}[
    font=\sffamily, 
    market_box/.style={
        rectangle, rounded corners=4pt,
        draw=box_color!80!black, fill=box_color, text=white,
        align=center, minimum height=1.3cm, minimum width=4.5cm,
        drop shadow={opacity=0.2, shadow xshift=1pt, shadow yshift=-1pt}, 
        font=\sffamily\bfseries\small
    },
    us_box/.style={
        rectangle, rounded corners=4pt,
        draw=bridge_red, fill=white, text=bridge_red, line width=1.5pt,
        align=center, minimum height=1.6cm, minimum width=5.5cm, 
        drop shadow={opacity=0.2, shadow xshift=1pt, shadow yshift=-1pt}, 
        font=\sffamily\bfseries\small 
    },
    interface_box/.style={
        rectangle, rounded corners=4pt,
        draw=black!70, fill=black!10, text=black!90, line width=1pt,
        align=center, minimum height=1.2cm, minimum width=4.5cm,
        drop shadow={opacity=0.2, shadow xshift=1pt, shadow yshift=-1pt},
        font=\sffamily\bfseries\small
    },
    flow_arrow/.style={
        ->, >=stealth, line width=1.4pt, color=arrow_color, rounded corners=4pt
    },
    data_arrow/.style={
        ->, >=stealth, line width=1.2pt, color=arrow_color, dashed, rounded corners=4pt
    },
    lbl/.style={
        inner sep=3pt, font=\sffamily\small, text=black!90, align=center
    }
]

\node[market_box] (update) at (0, 8.5) {Market Update\\[-4pt]Arrives};
\node[us_box] (reopt) at (0, 6.5) {High Frequency DP\\[-4pt]Reoptimizes Schedule};
\node[interface_box] (submit) at (0, 4) {Submit Orders\\[-4pt]to Market};
\node[interface_box] (wait) at (0, 2) {Waiting Duration\\[-4pt](Latency $\Delta t$)};
\node[market_box] (match) at (0, 0) {Orders Matched\\[-4pt]or Missed};

\draw[flow_arrow] (update) -- (reopt);
\draw[flow_arrow] (reopt) -- (submit) node[midway, right=2pt, lbl] {Profitable trades\\[-4pt]identified};
\draw[flow_arrow] (submit) -- (wait);
\draw[flow_arrow] (wait) -- (match);

\draw[flow_arrow] (match.west) -- ++(-1.0,0) coordinate (left_turn) 
    -- (left_turn |- update.west) node[midway, left=2pt, lbl] {Updated Battery\\[-4pt]Schedule} 
    -- (update.west);

\draw[data_arrow] (reopt.east) -- ++(0.5,0) coordinate (right_turn) 
    -- (right_turn |- update.east) node[midway, right=2pt, lbl] {No trades\\[-4pt]required} 
    -- (update.east);

\end{tikzpicture}}
    \caption{Flowchart of our method's execution logic.}
    \label{fig:flowchart}
\end{figure}

Finally, we remark that when backtesting, we correctly account for changes in the order book caused by our actions. This in particular means that EPEX orders cleared against the orders of our policy are available shorter than in the actual history, and consequently that subsequent clearing of orders by other market participants changes. This introduces a strong path-dependency of varying trading policies.

Although we take great care to model the market correctly, a limitation of the backtesting experiments is that, by the very nature of our analysis and the available data, we cannot take into account the effect that the orders placed by the strategy would have had on the behavior of other market participants. However, this drawback is inherent in the idea of backtesting and cannot be easily corrected.

All simulations were implemented in C\texttt{++} and run on a single AMD EPYC 9654 CPU, with the MILP solved using Gurobi v11.0.3 for C\texttt{++} \citep{gurobi}. Due to the simple nature of each single intrinsic optimization, solving the MILP on multiple threads does not lead to faster simulation times. Additionally, to provide a fair comparison, we set our MILP solver's time-limit to 2 seconds (MIP gap 10\%) and warmstart each optimization with the null-action, i.e. the solution of not changing the current battery schedule.  We publish our code in the Python package \textit{BitePy}, which allows users to easily run our DP simulations over a Python interface. See \ref{sec:code} for details. 

\subsection{Comparison between MILP and DP Rolling Intrinsic} \label{ssec:dp_vs_milp}
To evaluate the effectiveness of solving the intrinsic rolling problem using the DP approach instead of the MILP approach, we start by comparing the two methods. The MILP consistently provides the exact solution to the optimization problem \eqref{eq:intrinsic_vanilla}, while the DP solution generally slightly diverges. Note that, due to the RI's myopic nature and the path dependence of differing solutions, following the exact solution does not necessarily guarantee larger profits. For this reason, although inexact in solving single intrinsic problems, the DP-based rolling intrinsic strategy can be more profitable than the MILP formulation, as evidenced by the February and April results in Table \ref{tab:milp_dp_artifical}.

We choose the most clinical setting for our comparison, by artificially fixing the solve time and technical delay to $0$ ms for both approaches. This results in a policy that ignores technical delays and solves the intrinsic optimization for each relevant order book update in such a way that the resulting all-or-none limit orders are matched instantaneously with the exchange after each solve.

Table \ref{tab:milp_dp_artifical} reports this comparison for the first week of February, April, July, and October 2021. All three DP methods use a constant equidistant grid $G$, with storage discretizations $m=101$, $m=51$ and $m=11$. In summary, we can say that even without delays, the DP achieves comparable rewards to MILP and significantly outperforms it in terms of runtime by a factor of 100-1000 for the coarser storage discretization $m=11$. While the DP simulation times remain constant, MILP simulation times vary, even though the number of intrinsic solves is almost equal, as the complexity of the MILP problem is highly influenced by current market conditions. Interestingly, a finer storage discretization does not provide a clear advantage over coarser discretizations in this experiment, but shows an increase in profits for our yearly simulation results, as evidenced by the following sections.

Bootstrapping the test results reveals no significant difference ($p\approx0.66$) in the hourly (per-product) rewards generated by the MILP compared to our DP methods.

Ultimately, the proposed DP method trades a mathematically negligible single-step approximation error for a computational speedup of up to three orders of magnitude, all while maintaining the exact same long-term financial performance as the traditional exact MILP approach.

\begin{table}[ht]
    \centering
    \footnotesize

    \renewcommand{\arraystretch}{0.85} 
    \begin{tabular}{c lrrrrr} 
    \toprule
     & 
       & \textbf{Reward} & \textbf{Sim-Time} & \textbf{Cycles/Day} & \textbf{Solves} & \textbf{Traded Vol.} \\
     & 
       & [\texteuro] & [h]  & & & [MWh]  \\
    \midrule
    \multirow{4}{*}{\rotatebox{90}{\textbf{\begin{tabular}{@{}c@{}}Feb\\[-2pt]01-14\end{tabular}}}}
      & MILP & 7895 & 6.77 & 2.3 & 959900 & 1901 \\
      & DP-101 & 7538 & 0.43 & 2.3 & 959900 & 1825 \\
      & DP-51 & 8051 & 0.22 & 2.2 & 959800 & 1905 \\
      & DP-11 & 7851 & 0.05 & 2.2 & 959500 & 1878 \\[0.5pt] 
    \cmidrule(lr){2-7} 
    \addlinespace[0.5pt] 
    \multirow{4}{*}{\rotatebox{90}{\textbf{\begin{tabular}{@{}c@{}}Apr\\[-2pt]01-14\end{tabular}}}}
      & MILP & 9487 & 15.81 & 2.5 & 733000 & 1856 \\
      & DP-101 & 9434 & 0.34 & 2.5 & 733100 & 1817 \\
      & DP-51 & 9353 & 0.17 & 2.5 & 733000 & 1792 \\
      & DP-11 & 9543 & 0.04 & 2.5 & 733200 & 1872 \\[0.5pt] 
    \cmidrule(lr){2-7} 
    \addlinespace[0.5pt] 
    \multirow{4}{*}{\rotatebox{90}{\textbf{\begin{tabular}{@{}c@{}}Jul\\[-2pt]01-14\end{tabular}}}}
      & MILP & 11565 & 17.07 & 2 & 901000 & 1948 \\
      & DP-101 & 11161 & 0.43 & 2 & 901000 & 1873 \\
      & DP-51 & 11158 & 0.22 & 2 & 900900 & 1884 \\
      & DP-11 & 11181 & 0.05 & 2 & 901000 & 1899 \\[0.5pt] 
    \cmidrule(lr){2-7} 
    \addlinespace[0.5pt] 
    \multirow{4}{*}{\rotatebox{90}{\textbf{\begin{tabular}{@{}c@{}}Oct\\[-2pt]01-14\end{tabular}}}}
      & MILP & 28553 & 56.93 & 2.6 & 1012900 & 3350 \\
      & DP-101 & 28134 & 0.47 & 2.6 & 1012500 & 3231 \\
      & DP-51 & 27952 & 0.24 & 2.5 & 1012600 & 3190 \\
      & DP-11 & 28014 & 0.06 & 2.5 & 1012600 & 3236 \\ 
    \bottomrule
    \end{tabular}%
    \caption{Performance comparison for four representative weeks in 2021 between the rolling intrinsic solved using the benchmark MILP and three DP approaches with different approximations due to their value function granularities. All results are rounded to significance.}
    \label{tab:milp_dp_artifical}
\end{table}

\subsection{Yearly High-Frequency Rolling Intrinsic Trading}
\label{ssec:yearly_comp}
\begin{figure}[t]
    \centering
    \includegraphics[width=1.0\textwidth]{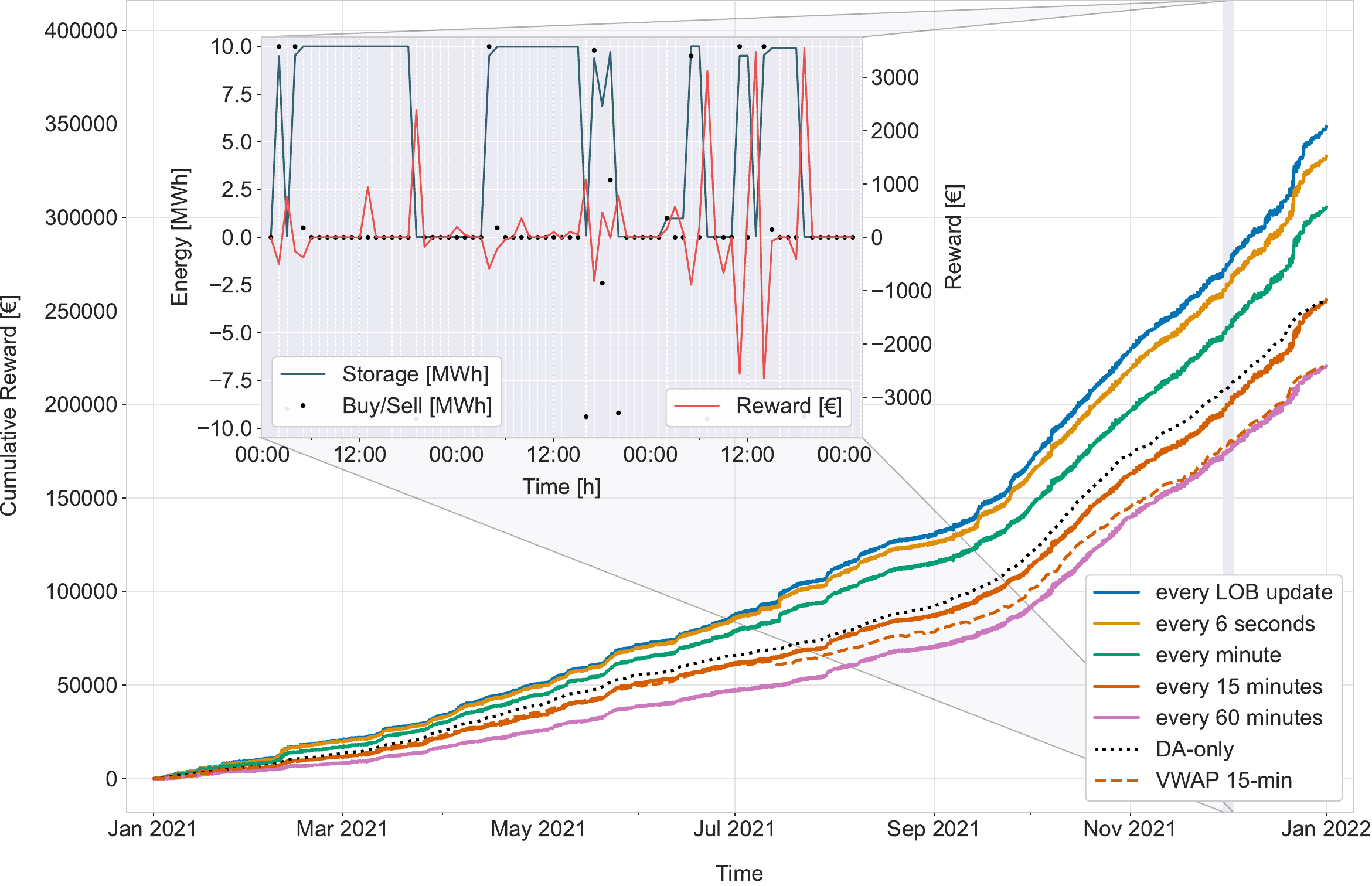}
    \caption{The cumulative reward of the intrinsic rolling over the entire year 2021 for various intrinsic-solve frequencies, compared against two benchmark strategies. The inset shows the battery operation over the span of three arbitrarily selected days for the strategy that solves the intrinsic at every relevant LOB update (taking into account technical delays), with the final schedule determined by the RI.}
    \label{fig:year_std_result}
\end{figure}
Leveraging the speed of our proposed DP method, we present full-year DP-RI trading results under a standard battery and simulation setup for the year 2021 in Figure \ref{fig:year_std_result}, and compare them against two materially different, implementable industry benchmarks: The first is a static Day-Ahead (DA) benchmark, which optimizes the battery schedule once per day based on the cleared DA auction prices. The second is an implementable continuous intraday benchmark based on Volume-Weighted Average Prices (VWAPs). Adapted from a state-of-the-art industry battery index \citep{isea_battery_revenue_index}, this strategy calculates VWAPs from transaction data in trailing 15-minute intervals and optimizes the schedule based on these aggregated price curves at that frequency. To ensure a fair and consistent multi-day evaluation, we adapted the published index by optimizing across all actively tradable products without artificial end-of-day state-of-charge constraints, integrating our linear degradation costs, and utilizing an MILP solver to correctly handle negative prices and enforce minimum bid-size increments. To make this strategy implementable on a continuous market, we submit the resulting target positions as market orders to the actual historical limit order book at the time of solution (starting trading at 11 pm for the next day's delivery to avoid the extreme illiquidity present at the 3 pm gate opening).

Figure \ref{fig:year_std_result} illustrates the cumulative rewards of these benchmarks against our proposed LOB-based strategy. When the VWAP strategy is forced to interact with the actual LOB, rather than assuming infinite liquidity at the theoretical average price, it systematically overestimates available volumes and suffers from significant price slippage upon execution. While the LOB strategy re-optimizing every 15 minutes yields profits roughly comparable to the static DA benchmark, our high-frequency LOB strategies demonstrate a clear performance advantage.. By accurately modeling depth-dependent pricing and capturing short-lived opportunities through millisecond-level reoptimization, our fastest LOB strategy outperforms the realistically implemented VWAP benchmark by more than 36\% in realized annual profits, highlighting the critical value of detailed microstructural modeling of the continuous intraday market.

The DP settings used for the comparisons are outlined in Section \ref{ssec:case_study_setting} with a value function discretization setting of $m = 11$ and a technical delay of $200$ ms plus measured actual solve times. We evaluate performance starting from the baseline with a solve-frequency of $60$ minutes, i.e. solving the intrinsic optimization once every hour, which results in a reward of \texteuro{$221$k} (8.7k orders submitted at the exchange) and a simulation runtime of $5.6$ minutes, compared to the strategy that solves it with every relevant order book update (apart from the waiting durations introduced previously) yielding a reward of \texteuro{$349$k} (30k orders submitted at the exchange) and a simulation runtime of $86$ minutes. We define any order message submitted to the LOB as relevant, if it is placed at the head of its corresponding bid/ask stack. These results clearly demonstrate that solving the intrinsic at higher frequencies leads to progressively increasing profits as the trading frequency rises. In this case, a $58\%$ increase between our slowest and quickest solve frequency. We observe a clear increase in the slope of the cumulative reward starting in October 2021. This most likely originates from higher and more volatile gas prices in Europe during this time, which subsequently translates to more extreme intraday electricity prices, which increases the revenue potential of storage.

To extend the information presented in Figure \ref{fig:year_std_result}'s inset, Figure \ref{fig:inset_decisions} depicts the detailed sequential decisions made by our policy on the last day of the inset. Looking closely at the overlapping dots, one can observe that some trading decisions are only a few hundred milliseconds apart.

\begin{figure}[h!]
    \centering
    \includegraphics[width=\textwidth]{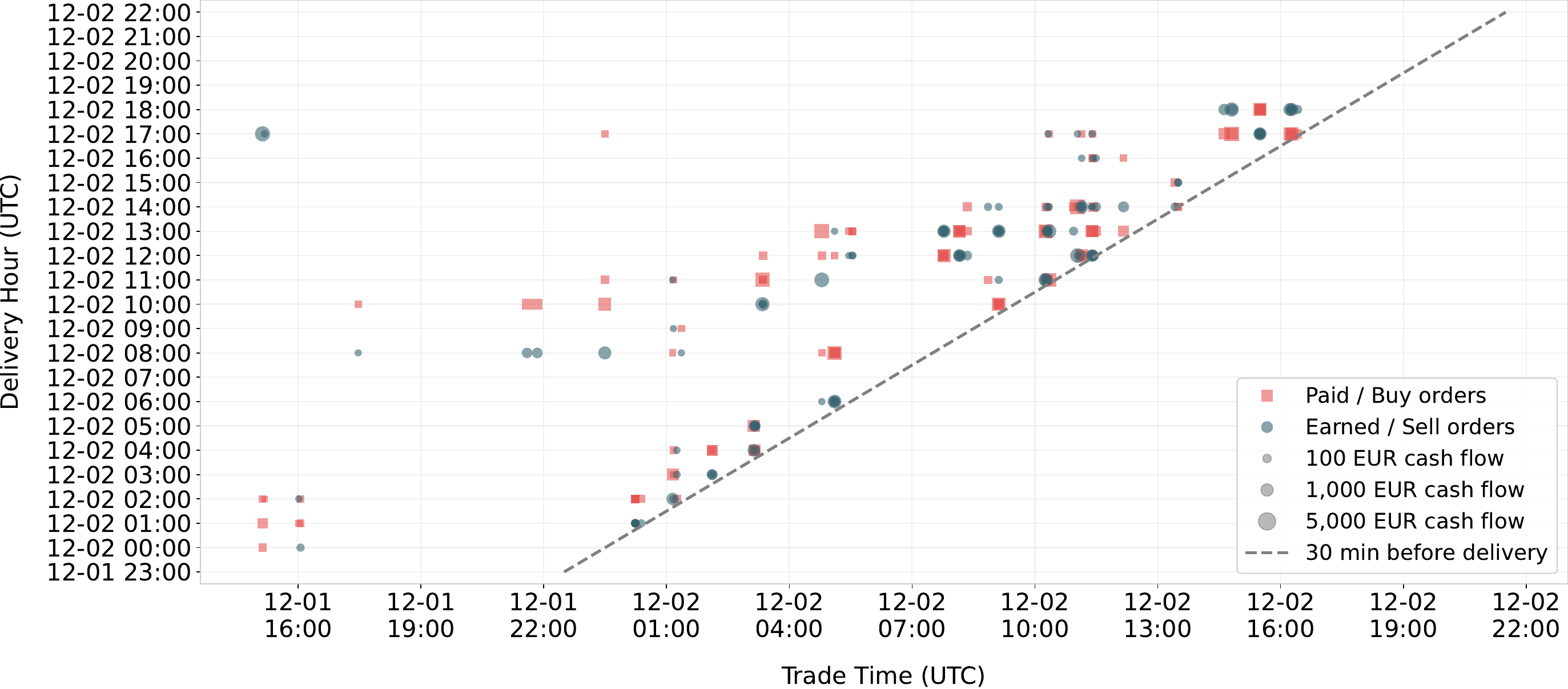}
    \caption{Detailed trading decisions for products traded on 2021-12-02 (Berlin time), which is the last day depicted in Figure \ref{fig:year_std_result}'s inset.}
    \label{fig:inset_decisions}
\end{figure}

An analysis of the daily profit distribution reveals a highly favorable, asymmetric risk profile inherent to the myopic rolling intrinsic policy. The algorithm effectively capitalizes on extreme scenarios characterized by high price volatility, resulting in a heavy right tail of highly profitable days (most notably in late 2021). Conversely, the strategy is virtually immune to extreme daily losses. During periods of low liquidity or unfavorable prices, the policy simply refrains from trading rather than speculating.

Figures \ref{fig:year_heatmap} and \ref{fig:execOrders_distance} illustrate DP-RI decisions, using the same settings and our quickest strategy, solving the intrinsic at every relevant LOB update. Figure \ref{fig:year_heatmap} provides an overview of the annual battery state-of-charge (SoC) schedule shaped by the RI’s trading decisions, along with the daily trading volumes. We observe a distinct average SoC pattern, suggesting a strong daily periodicity of arbitrage opportunities, with higher charge levels during the early morning and afternoon hours. This can be seen by looking at the trading decisions in the central heat map of trading decisions as well as the mean SoC profile. 

The lower plot shows the frequency of different SoC levels throughout the year. Given that our storage has a duration of one hour, one would expect the storage to be either completely full or completely empty most of the time. However, the SoC occupies levels in between these extremes surprisingly often. This is most likely caused mainly by complex trading behavior induced by nonlinearities in order book-based trading, which makes simple bang-bang strategies suboptimal. This highlights the difference between a naive price-taking implementation of the rolling intrinsic and our more realistic approach.  

Lastly, when looking at the daily traded energy over the year on the top of the figure, we see increased trading activity towards the end of the year, which fits the increased slope of cumulative profits for this time period observed in Figure \ref{fig:year_std_result}.

\begin{figure}[h!]
    \centering
    \includegraphics[width=1.0\textwidth]{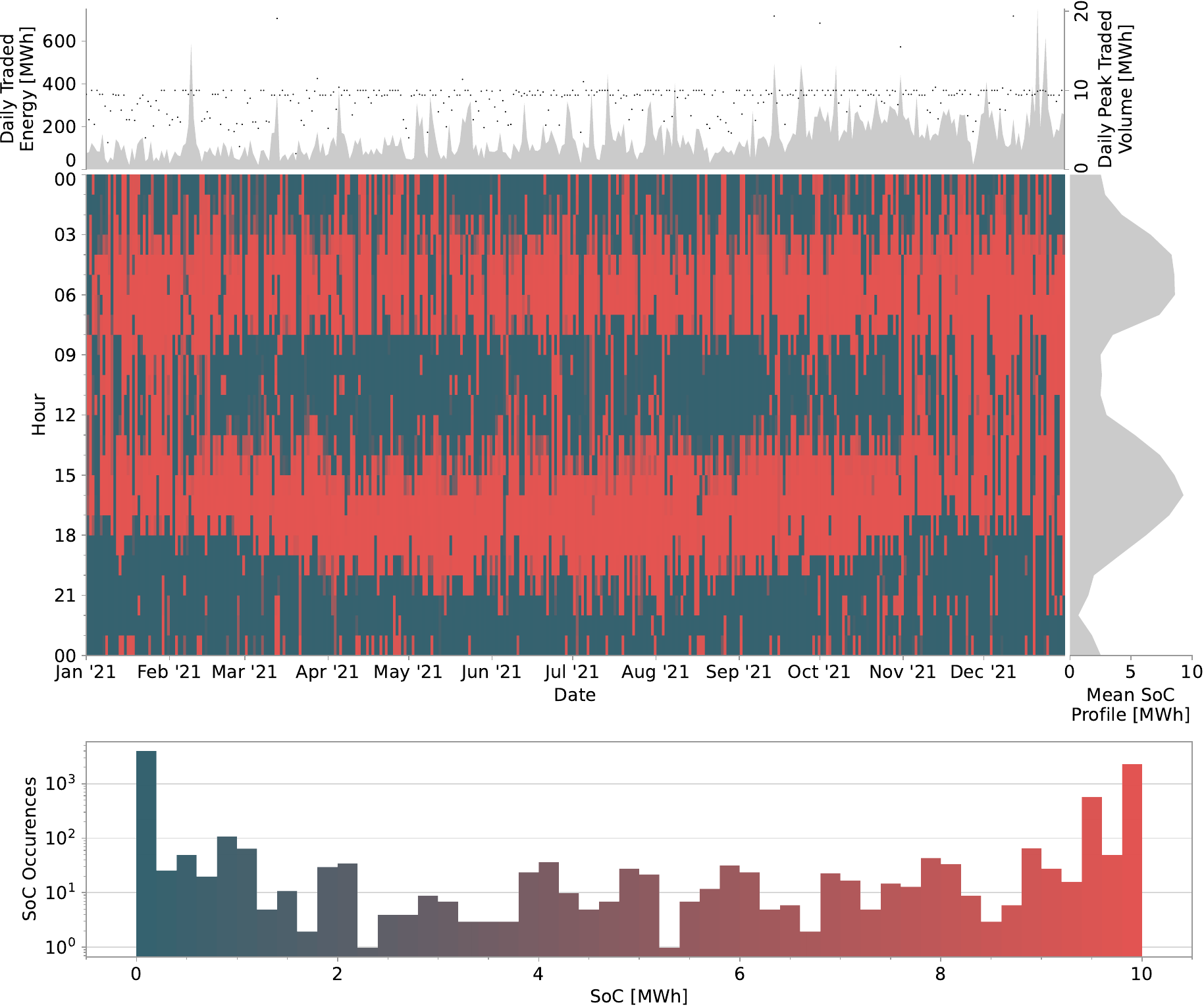}
    \caption{Yearly heatmap of the final battery schedule for the year 2021 determined by the RI. The upper plot shows the final state of charge (SoC) of the battery for each hour of the year, with the color gradient defined in the lower plot. The gray areas show the total traded energy per day and the yearly mean SoC profile (avg. SoC for each hour 1-24 of the day across the entire year), while the dots in the upper inset represent the largest trade per day submitted by the RI. We observe a clear preference for charging the battery in the early morning and afternoon hours. The lower plot shows the logarithmic distribution of SoC states of the battery, as determined by the RI's schedule. Figure adapted from \citep{brudermueller2023smart}.}
    \label{fig:year_heatmap}
\end{figure}

Figure \ref{fig:execOrders_distance} shows the number of orders executed by the RI, classified by their proximity to physical delivery. Comparing this to Figure \ref{fig:liquidity} showing the number of orders placed in the market, we see that the pattern is much more uniform. This is likely due to the fact that the rolling intrinsic tries to capitalize on time spreads which often induces trading pairs of products with large differences in the time to delivery, e.g., when buying cheap in the early morning to sell expensive in the evening. This induces the RI to frequently trade products that are not yet liquid, which is one of the biggest weaknesses induced by the myopic nature of the policy.

\begin{figure}[t]
    \centering
    \includegraphics[width=0.6\textwidth]{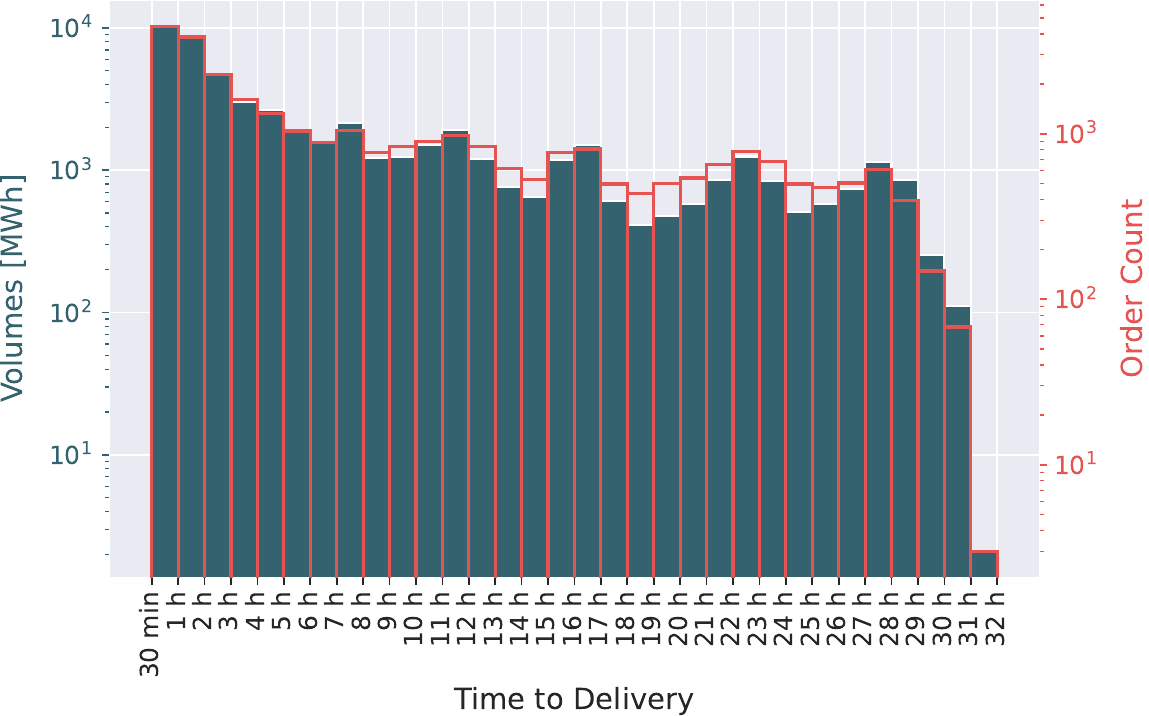}
    \caption{Logarithmic order submission-time analysis of the RI simulated over the entire year 2021.}
    \label{fig:execOrders_distance}
\end{figure}

Finally, as discussed in Section \ref{ssec:case_study_setting}, market latency can lead to execution failures for our submitted limit orders. When an order is missed because the order book has already shifted during the combined solve time and technical delay, the algorithm may temporarily record a non-physical position that it must correct in the immediately following optimization step. 
Table \ref{tab:simulation_results_missed} details the frequency, time-to-resolution, and financial impact of these feasibility corrections over the full-year simulation. Across different dynamic programming discretizations, approximately 7\% of the orders were missed due to latency. The rolling policy naturally resolved these infeasibilities rapidly, with a median resolution time of 0.3 seconds. However, these missed opportunities and subsequent corrections carry a tangible cost, reducing the total theoretical simulation profits by roughly 6\% to 7\%. The consistency of these statistics across different state grid sizes ($m=|G|$) aligns with expectations, as the measured solve time for all discretizations remains far below the fixed 200 ms technical delay, making the latter the main reason for missing orders. 
Isolating the exact financial penalty of a missed order is non-trivial, as subsequent re-optimizations might replace the order, abandon it, or combine it with new trades, sometimes even at more favorable prices. To estimate the foregone profit impact reported in Table \ref{tab:simulation_results_missed}, we employ a volume-proportional scaling approximation. If an order fails to match, we examine the replacement orders placed in the immediate next solve that share the same trade direction. We then scale the cost comparison by the volume share of the missed order. For example, if the algorithm misses a 1.0 MWh buy order at 10 \texteuro/MWh, and the subsequent solve proposes a 1.3 MWh buy order at 11 \texteuro/MWh, we estimate the correction loss based on the original 1.0 MWh volume: $((1.3 \times 11)\frac{1.0}{1.3} - 1.0 \times 10) = 1$ \texteuro.

This method provides a realistic approximation of the friction introduced by execution latency without attempting to track ambiguous cascading correction-intents across multiple future solves. We acknowledge, however, that this volume-proportional scaling is a heuristic that introduces a degree of uncertainty. Due to the highly path-dependent nature of the rolling intrinsic policy, a missed order does not simply force a 1-to-1 replacement trade. It can occasionally cause the algorithm to abandon a trade entirely or shift its strategy to different delivery hours. Consequently, the true opportunity cost of a missed order could be slightly higher or lower than our estimate. Nevertheless, because this approximated latency penalty, which accounts for a roughly 6\% - 7\% reduction in total theoretical profit, is already fully incorporated into our reported trading results, this uncertainty does not threaten the robustness of our core conclusions. Even when absorbing these estimated latency frictions, our high-frequency DP approach continues to significantly outperform the lower-frequency and day-ahead benchmarks.

\begin{table}[h!]
    \footnotesize
    \centering
    \begin{tabular}{lrrrrrr}
            \toprule
             & \multicolumn{2}{c}{\textbf{Corrections}} & & & \multicolumn{2}{c}{\textbf{Time (sec)}} \\
            \cmidrule(lr){2-3} \cmidrule(lr){6-7}
            \textbf{Type} & Count & \% of solves & Costs (\texteuro) / \% of total & Sim Reward (\texteuro) & Mean & Median \\
            \midrule
            DP-11  & 867 & 7.3\% & $-25,200$ / $-7.2$\% & $349,000$ & $0.93$ & $0.30$ \\
            DP-51  & 996 & 7.1\% & $-21,400$ / $-6.1$\% & $351,000$ & $8.61$ & $0.30$ \\
            DP-101 & 989 & 6.9\% & $-24,000$ / $-6.8$\%  & $352,000$ & $1.96$ & $0.30$ \\
            \bottomrule
        \end{tabular}
    \caption{Statistics on missed orders and subsequent feasibility corrections over the year 2021. 'Costs' represent the foregone profit impact of correcting non-physical positions caused by execution latency.}
    \label{tab:simulation_results_missed}
\end{table}

\subsection{Parameter Evaluation}\label{ssec:param_eval} 
The DP rolling intrinsic optimization relies on numerous parameters, each influencing trading behavior and solution accuracy in different ways. This section compares key storage and optimization parameters, demonstrating that our year-long trading results remain robust and that the parameters’ effects align with expectations.

To achieve this, we compare the effects of varying injection-efficiencies (assuming $\eta^+ = \eta^-$), linear degradation costs $\nu_\text{deg}$ and maximum storage capacities $\bar s$, each for three different equidistant DP value function grids $G$ with granularity $\vert G \vert = m$. Table \ref{tab:dp_param} provides a shortened overview of yearly rewards for each of the optimizations. Two clear trends can be observed as expected: A finer grid $G$ (larger $m$) results in more profitable trading as the solution is more accurate, and more penalizing injection efficiencies $\eta$ and linear degradation costs $\nu$ result in lower rewards because fewer trading opportunities become profitable. Both parameters have a significant effect on profits, with the most favorable settings nearly doubling the profits of the least favorable setting. This is especially interesting for the case of degradation costs, which are likely to decrease significantly in the near future, thereby greatly increasing the profitability of short-term trading on intraday markets.

Additionally, batteries with a smaller maximum storage capacity produce larger profits per MWh storage capacity, which is mainly caused by limited market depth and liquidity. Operating a larger storage system requires executing larger trade volumes. To fill these larger volumes, the algorithm must sweep deeper into the limit order book, consequently accepting progressively worse prices and shrinking the captured price spreads. Similarly, a battery with a duration of $2$h yields lower profits, but more than half of the profits observed for our standard setup with a duration of $1$h, because it can no longer charge and discharge in the hour(s) of the singular best price(s) of the day.

\begin{table}[t]
    \centering
    \footnotesize
    \begin{minipage}[t]{0.48\linewidth} 
        \centering
        \begin{tabular}{lrrrr}
        \toprule
        & \multicolumn{4}{c}{\boldmath$\eta^+ = \eta^-$}  \\
        \cmidrule(lr){2-5}
        Reward [\texteuro]     & 0.9  & 0.95  & 0.99 & 1.0  \\
        \midrule
        \textbf{DP-11} & 277000 & 349000 & 441000 & 478000 \\
        \textbf{DP-51} & 278000 & 351000 & 444000 & 482000 \\
        \textbf{DP-101} & 279000 & 352000 & 442000 & 483000 \\
        \midrule
        & \multicolumn{4}{c}{\boldmath$\nu_\textbf{deg}$ \textbf{[€/MWh]}} \\
        \cmidrule(lr){2-5}
        Reward [\texteuro]      & 8 & 4 & 2 & 0 \\
        \midrule
        \textbf{DP-11} & 263000 & 349000 & 416000 & 434000 \\
        \textbf{DP-51} & 264000 & 351000 & 419000 & 439000 \\
        \textbf{DP-101} & 265000 & 352000 & 419000 & 441000 \\
        \bottomrule
        \end{tabular}
    \end{minipage}
    \hfill 
    \begin{minipage}[t]{0.48\linewidth} 
        \centering
        \begin{tabular}{lrrrr}
        \toprule
         & \multicolumn{4}{c}{\boldmath$\bar{s}$ \textbf{[MWh]} \ 1h-battery} \\
        \cmidrule(lr){2-5}
        Reward/$\bar{s}$ $[\frac{\text{\texteuro}}{\text{MWh}}]$  & 5 & 10 & 20 & 40 \\
        \midrule
        \textbf{DP-11} & 36600 & 34900 & 32600 & 29700 \\
        \textbf{DP-51} & 36900 & 35100 & 33200 & 30400 \\
        \textbf{DP-101} & 36900 & 35200 & 33400 & 30600 \\
        \midrule
        & \multicolumn{4}{c}{\boldmath$\bar{s}$ \textbf{[MWh]} \ 2h-battery} \\
        \cmidrule(lr){2-5}
        Reward/$\bar{s}$ $[\frac{\text{\texteuro}}{\text{MWh}}]$  & 5 & 10 & 20 & 40 \\
        \midrule
        \textbf{DP-11} & 27300 & 27000 & 25700 & 23900 \\
        \textbf{DP-51} & 27100 & 27400 & 26200 & 24600 \\
        \textbf{DP-101} & 27700 & 27300 & 26400 & 24700 \\
        \bottomrule
        \end{tabular}
    \end{minipage}
\caption{A reward comparison of simulation results for the full year 2021, using four different storage parameter settings. All simulation base-settings are equal to the setting defined in \ref{ssec:case_study_setting}, setting the ping-delay to 200 ms and using actual measured solve-times at execution. Varying storage parameters has the expected effect on yearly rewards, where a finer DP storage discretization generally leads to marginally higher rewards, increasing degradation costs and losses reduces the rewards, and operating a larger storage yields lower rewards per storage capacity. All results are rounded to significance, and the specific battery reward on the right table is given in units of \texteuro/MWh.}
\label{tab:dp_param}
\end{table}

\subsection{Parametrization of the Rolling Intrinsic}
The rolling intrinsic strategy is suboptimal due to its myopic nature. In particular, the lack of foresight leads to \emph{impatient trading decisions} that use flexibility early on, thereby foregoing parts of the storage's revenue potential. This happens in particular when the RI trades products with limited liquidity, typically long before gate closure, and for rather disadvantageous prices.

This opens the door for parametric variations of the rolling intrinsic policy that aim to improve on these issues by adjusting the trading behavior. However, finding optimal parameters requires extensive tuning efforts and, therefore, only works in combination with an efficient implementation of the policy. Our method therefore opens the door to optimized variations of the rolling intrinsic, which can be trained within a realistic timeframe.

To illustrate this point, we introduce an empirically motivated unitless linear parameter $\phi$ to the intrinsic optimization. It acts as a penalty for products with a large bid-ask spread. We then adapt the optimization \eqref{eq:intrinsic_dp} by extending $\pi_t$ to $\hat{\pi}_t = \pi_t - (\phi \cdot \delta_{t^*, t}) \mid f_t^0 - f_t \mid$, with $\delta_{t^*, t}$ defined as the bid-ask spread of product $t$ at time $t^*$. Intuitively, it becomes clear that $\phi$ incentivizes the rolling intrinsic to trade at times of higher market liquidity, i.e. at times closer to delivery, to reduce the chance of trading less profitably too early.

The increased speed of our DP solution method allows us to quickly search for optimal parameters $\phi$. For the sake of simplicity and just to show a proof of concept, we fix the parameters $\phi$ at constant values for each trading month, and use the previous month $m-1$ to train $\phi$ for a given month $m$. When evaluating this strategy over the entire year 2021, this sliding window parameter training results in $12$ separate training windows. We train the parameters on the highest intrinsic-solve frequency (see section \ref{ssec:yearly_comp}) using DP and simulation settings outlined in Section \ref{ssec:case_study_setting}, which represents the most realistic training scenario. We treat our monthly simulation rewards as a black box function $f(\phi)$ and, given its simple dependence on a singular parameter, use Brent's method to find its maximum.

The resulting optimal parameters for the year 2021 have an average penalty of $\bar{\phi} = 2.1 \pm 1.6$. Comparing year-long out-of-sample results of the newly trained policy yields an increase of $8.4$\% in trading reward over the entire year 2021, accumulating \texteuro{$376$k}, compared to the \texteuro{$347$k} reward of the standard policy with the penalty set to $0$, i.e., $\phi \equiv 0$.

While this specific extension remains an illustrative tuning exercise rather than a comprehensive new trading framework, it highlights the immense practical value of the proposed fast DP approximation. Ultimately, it demonstrates that our method provides the computational foundation required to rapidly prototype, iteratively train, and backtest more sophisticated parametric or machine-learning-driven strategies over massive, high-frequency datasets, a task that would be strictly impossible using traditional formulations.


\section{\label{sec:conclusion}Conclusion}
In this paper, we present a novel method to compute the rolling intrinsic policy for trading with a battery on continuous intraday markets. Our method achieves a solution speed improvement of up to three orders of magnitude while maintaining a sufficient level of accuracy. Our work lays the foundation for extensive future research through the extension and adaptation of our rolling intrinsic formulation and implementation, both for researchers and industry.

Our findings demonstrate the importance of strategies that take into account the exact workings of the continuous intraday market and, at the same time, can be executed fast enough to keep up with the millisecond pace of the market. We show through realistic backtesting that a solution frequency of at least every second increases rolling intrinsic profits substantially compared to solving the intrinsic only every few minutes or even every hour. Our synthesis of the backtesting results underscores a fundamental market paradigm: the continuous clearing mechanism structurally rewards high-frequency execution. Strategies that fail to react to tick-level order book updates frequently leave value on the table, as delayed optimization limits a battery's ability to lock in short-lived price spreads.

Furthermore, our efficient implementation enables us to backtest the profitability of the rolling intrinsic policy over extended time periods, which requires a substantial number of intrinsic solves across the simulation period, making traditional MILP or LP methods impractical due to their comparatively slow solution time. This also makes it possible to train parametric extensions of the rolling intrinsic that correct some of its shortcomings, as we demonstrate by optimizing a penalty for trading illiquid products, which substantially increases trading profits.

For practitioners, our results illustrate that resolving the limit order book at high frequencies is a highly relevant factor when designing and backtesting intraday trading algorithms. Our empirical analysis focuses exclusively on the 2021 German continuous intraday market, which operates within the European Single Intraday Coupling. While absolute profitability will naturally vary across different historical years, market regimes, and specific international regulatory environments, our findings clearly highlight the structural mechanics of continuous clearing. Within this specific market design, strategies equipped with the computational speed to process and react to tick-level data possess a distinct execution advantage over slower, lower-frequency benchmarks.

We acknowledge specific limitations in our modeling approach, particularly the linear approximation of battery degradation and the discretization error inherent in the DP formulation relative to the exact MILP. Furthermore, while our use of full-depth order book data ensures high fidelity, the backtest remains a simulation. It inherently cannot capture endogenous market impact. Specifically, how other market participants might have dynamically adjusted their strategies in response to our orders. Lastly, a notable limitation of our current approach is the exclusive focus on the intraday market, which omits the co-optimization potential with other highly profitable platforms such as the day-ahead or frequency containment reserve markets.

An important future extension of our work lies in more sophisticated parametrizations of the rolling intrinsic. Our naive method could easily be adapted to multiple additional parameters like trading delays, affecting various aspects of the RI's trading behavior, further utilizing the speed of our optimization method. Additionally, incorporating nuanced degradation costs, such as State of Health (SoH) and non-linear depth-of-discharge, would significantly enhance the realism and applicability for practitioners.  Furthermore, future work could naturally extend the trading to a multi-market scenario, where additional decision variables expand the scope of optimization to multiple subsequent electricity markets, e.g. the reserve or day-ahead markets.

Finally, extending our intrinsic optimization with dynamic price forecasts, or more general value functions of future actions, would greatly improve the RI's profitability, as the optimization would be able to anticipate future price changes and capture the option value of waiting, rather than relying on strictly myopic execution. Fine-tuning non-myopic, forecast-driven, or stochastic models to operate at the high frequencies enabled by our DP approach is a logical and highly promising avenue for future research.

\section*{Declaration of generative AI and AI-assisted technologies in the writing process}
During the preparation of this work, the authors used Gemini and ChatGPT in order to improve language and readability. After using this tool/service, the authors reviewed and edited the content as needed and take full responsibility for the content of the published article.

\bibliographystyle{elsarticle-harv} 
\bibliography{bib}

\appendix
\section{Nomenclature}

\noindent
\begin{longtable}{@{}ll@{}}
\multicolumn{2}{@{}l}{\textbf{Sets and Indices}} \\[0.5ex]
$\mathcal{T}$ & Set of delivery time periods / futures contracts \\
$t$ & Index for time period / delivery contract \\
$t^*$ & Current time / decision time \\
$t_0$ & Time when the earliest tradable contract goes into delivery \\
$\mathcal{O}_t$ & Set of all limit orders in the order book for contract $t$ \\
$\mathcal{O}_t^+, \mathcal{O}_t^-$ & Set of currently active asks (+) and bids (-) for contract $t$ \\
$i, j$ & Indices for individual limit orders \\
$G$ & Discretized finite grid of possible storage states \\
\\[1ex]

\multicolumn{2}{@{}l}{\textbf{Parameters and Constants}} \\[0.5ex]
$\bar{s}$ & Maximum storage energy capacity [MWh] \\
$s_0$ & Initial storage level [MWh] \\
$\bar{f}, \underline{f}$ & Physical bounds for maximum injection and withdrawal [MWh] \\
$f_t^0$ & Initial future positions [MWh] \\
$P_i$ & Limit price of order $i$ [\texteuro/MWh] \\
$Q_i$ & Available quantity of order $i$ [MWh] \\
$\eta^+, \eta^-$ & Efficiency factors for charging (+) and discharging (-) \\
$\nu_{\text{trade}}$ & Volume-based trading cost [\texteuro/MWh] \\
$\nu_{\text{deg}}$ & Linear degradation cost [\texteuro/MWh] \\
$\nu$ & Combined variable cost ($\nu_{\text{trade}} + \nu_{\text{deg}}$) [\texteuro/MWh] \\
$u$ & Minimum tradable unit/quantity [MWh] \\
$m$ & Number of grid points in the discretized storage state grid $G$ \\
$\kappa$ & Action discretization step parameter \\
$\phi$ & Penalty parameter for trading products with large bid-ask spreads \\
$\delta_{t^*, t}$ & Bid-ask spread of product $t$ at time $t^*$ [\texteuro/MWh] \\
\\[1ex]

\multicolumn{2}{@{}l}{\textbf{Variables and Functions}} \\[0.5ex]
$s_t$ & Storage level at time $t$ [MWh] \\
$f_t$ & Future position for contract $t$ [MWh] \\
$f_t^+, f_t^-$ & Total to be bought (+) and sold (-) volumes for contract $t$ [MWh] \\
$i_t, w_t$ & Total bought (+) and sold (-) volumes for contract $t$ including previous positions [MWh] \\
$q_t, q_i$ & Traded quantity of a contract or specific order $i$ [MWh] \\
$k_i, k_t$ & Integer multipliers for minimum tradable units \\
$\alpha_t$ & Binary decision variable preventing simultaneous charge/discharge \\
$V_t(s_{t-1})$ & Dynamic programming value function at stage $t$ \\
$\tilde{V}_t(s_t)$ & Linearly interpolated approximation of the value function \\
$\pi_t$ & Function encoding current order book cost/revenue information \\
$S(s_{t-1}, f_t)$ & State transition function yielding the subsequent storage level \\
\end{longtable}

\section{Code Publication}
\label{sec:code}
We publish our method as a Python package \textit{BitePy} (Battery Intraday Trading Engine), where users can preprocess their intraday market data, set battery and DP parameters, run simulations and analyze results. It is hosted on PyPi and can be easily installed via \textit{pip install bitepy}. Detailed documentation and tutorials on the package can be found on \href{https://github.com/dschaurecker/bitepy}{GitHub}.

\section{Acknowledgements}
The authors thank Markus Kreft and Patrick Langer for the helpful methodological discussions and for support with some of the figures. In addition, we thank Simon Hirsch for the factual checks, general feedback and for the support in wrapping the Package to Python and setting up the documentation.

\end{document}